\documentclass[11pt]{article}

\usepackage{latexsym}
\usepackage{amsmath}
\usepackage{amssymb}
\usepackage{amsthm}
\usepackage{amscd}
\usepackage[mathscr]{eucal}
\usepackage{fullpage}
\usepackage{graphicx}
\usepackage{subfigure}
\usepackage{float}
\usepackage{psfrag}
\usepackage{rotating}
\usepackage{color}

\newcommand{\R}{\mathbb{R}}

\renewcommand{\thesection}{\Roman{section}}

\definecolor{ggreen}{cmyk}{0.7,     0,      0.9,      0}
\definecolor{viol}{cmyk}{0.3,1,0,0}
\definecolor{myred}{cmyk}{0.1, 1, 0.5, 0}
\definecolor{bblue}{rgb}{0.2, 0.29996, 0.8 }

\theoremstyle{plain}

\newtheorem{theorem}{Theorem}

\newtheorem{lemma}{Lemma}
\newtheorem{remark}{Remark}
\newtheorem{proposition}{Proposition}

\newcounter{mnotecount}[section]

\renewcommand{\themnotecount}{\thesection.\arabic{mnotecount}}

\newcommand{\mnote}[1]
{\protect{\stepcounter{mnotecount}}$^{\mbox{\footnotesize
			$
			\bullet$\themnotecount}}$ \marginpar{
		\raggedright\tiny\em
		$\!\!\!\!\!\!\,\bullet$\themnotecount: #1} }
\begin{document}
\title{\bf \huge{Global dynamics of Yang-Mills field and
perfect-fluid Robertson-Walker cosmologies}}

\author{Artur Alho$^{(1)}$, Vitor Bessa$^{(1,2,3)}$ and Filipe C. Mena$^{(1,2)}$\\\\
{\small $^{(1)}$Centro de An\'alise Matem\'atica, Geometria e Sistemas Din\^amicos,}
\\
{\small	Instituto Superior T\'ecnico, Universidade de Lisboa, Av. Rovisco Pais 1, 1049-001 Lisboa, Portugal}\\
{\small $^{(2)}$Centro de Matem\'atica, Universidade do Minho, 4710-057 Braga, Portugal}
\\
{\small $^{(3)}$Faculdade de Ci\^encias, Universidade do Porto, R. Campo Alegre, 4169-007 Porto, Portugal}
}
\maketitle
\begin{abstract}
We apply a new global dynamical systems formulation to flat Robertson-Walker cosmologies with a massless and massive Yang-Mills 
field 
and a perfect-fluid with linear equation of state as the matter sources. This allows us to give proofs concerning the global dynamics of the models including asymptotic source-dominance towards the past and future time directions. For the pure massless Yang-Mills field, we also contextualize well-known explicit solutions in a global (compact) state space picture.	
\end{abstract}
\section{Introduction}
%
The standard cosmological models are based on Einstein's theory of General Relativity and, in particular, on 4-dimensional Lorentzian manifolds, called {\em spacetimes}, with spatially homogeneous and isotropic metrics, called {\em Robertson-Walker metrics}. The dynamics on spacetime is governed by the {\em Einstein field equations} which, in general, are nonlinear partial differential equations but, under the assumptions of spatial homogeneity, turn into a system of nonlinear ordinary differential equations (ODEs). 

Many recent rigorous results about the dynamics of cosmological models result from the application of the theory of dynamical systems to ODE systems of Einstein field equations, see e.g.Refs \cite{WainElis} and \cite{Col03}. A common procedure in those approaches is to replace the metric variables, which form the system, with dimensionless variables using conformal rescalings  as well as suitable normalization factors, which intend to regularize the system on a compact state space. So, part of the problem is to understand how can this be achieved and how, within this framework, can one finally construct an autonomous dynamical system from physically relevant cosmological models.

Another potential complication in this problem are the matter fields considered in the Einstein field equations.  The typical matter content of a cosmological model is a perfect fluid with a linear equation of state \cite{WainElis}. However, since Yang-Mills fields describe the dynamics of elementary particles, they can also play an important role in the physics of the early universe, see e.g. Ref.\cite{pys-rep}. It is, therefore, of interest to obtain rigorous mathematical results for systems including such matter fields.  

So, given the importance of Yang-Mills fields in cosmology, particle physics and string theory, we revisit the problem of the dynamics of massless and massive Yang-Mills fields in co-evolution with a perfect fluid with a linear equation of state, on Robertson-Walker geometries. As we shall see, this results in the problem of analyzing a nonlinear ODE system of Einstein-Euler-Yang-Mills equations.  For the remaining of this section we shall, first, describe how this system is derived, then, explain the principal techniques we use comparing with previous approaches and, finally, summarize the main results of the paper.  

We then consider a spacetime manifold $(M,g)$, with Lorentzian metric $g$ satisfying the Einstein field equations,
\begin{equation}
\label{EFEs-general}
\text{Ric}-\frac{1}{2}S g=T
\end{equation}
where $\text{Ric}$ is the Ricci tensor and $S$ the scalar curvature of $(M,g)$, while $T$ is the energy-momentum tensor encoding the spacetime physical contents (see Ref.\cite{Choquet} for more details).

We assume that $(M,g)$ is spatially homogeneous and isotropic of the type $M=\R\times E^3/SO(3)$, where  the euclidean group $E^3$ is the isometry group of the spatial hypersurfaces. The most general form for $g$ which is invariant under the $E^3$ group is the flat Robertson-Walker (RW) metric, which in cartesian coordinates $(t,x,y,z)\in \R^+\times \R^3$, is given by
\begin{equation}
\label{metric}
g= -dt^2 + a^2(t)(dx^2+dy^2+dz^2),
\end{equation}
where $a$ is a $C^2$ positive function of (comoving) time $t$ called {\em scale factor}. 

We consider scalar fields as well as vector fields defined on $M$ which are compatible with our symmetry assumptions. 
In particular, we shall consider perfect fluid (scalar) matter with density $\rho_{\mathrm{m}}(t)$ and pressure $p_{\mathrm{m}}(t)$ and Yang-Mills (4-vector) fields $\vec A(t)$ with gauge group $G$. We write $A_\mu=A^a_{~\mu} L_a$, where greek indices $\mu,\nu,...=0,1,2,3$ denote spacetime indices, Latin indices $a, b,...=1,2,3$ are internal indices, and $L_a$ the infinitesimal generators of the Lie algebra associated to $G$ in the fundamental representation  $\text{Tr}[L^a L^b]=-\delta^{ab}/2$, where we use the minus sign convention.   

These two types of fields will be encoded in two tensors $T_{\mathrm{m}}$ and $T_{\mathrm{YM}}$ defined on $M$ such that
\begin{equation}
\label{Tmunugeral}
T_{\mu\nu}=T_{\mathrm{m}\mu\nu}+T_{\mathrm{YM}\mu\nu},
\end{equation}
Considering a globally defined timelike vector field $\vec u$ corresponding, physically, to the 4-velocity of the fluid, we may decompose $T_{\mathrm{m}\mu\nu}$ with respect to $\vec u$ as 
\begin{equation}
\label{TmunuPF}
T_{\mathrm{m}\mu\nu}=(\rho_{\mathrm{m}}+p_{\mathrm{m} }) u_\mu u_\nu+p_{\mathrm{m} }g_{\mu\nu},
\end{equation}
which must satisfy the Euler equations
\begin{equation}
\label{euler-eq}
\nabla^\mu T_{\mathrm{m}\mu\nu}=0.
\end{equation}
In our coordinate system, given $u^\mu=\delta^\mu_{~0}$, we simply get $T_{\mathrm{m}\mu\nu}=\text{diag} (\rho_\mathrm{m},p_\mathrm{m},p_\mathrm{m},p_\mathrm{m})$.
We shall further assume a linear equation of state 
$$p_\mathrm{m}=(\gamma_\mathrm{m}-1)\rho_\mathrm{m},$$
with $\rho_\mathrm{m}(t)\ge 0$ and the constant adiabatic index $\gamma_\mathrm{m}$ satisfying $0<\gamma_\mathrm{m}\leq2$, where $\gamma_\mathrm{m}=1$ corresponds to a pressureless (dust) fluid, $\gamma_\mathrm{m}=4/3$ to radiation, while the extreme values $\gamma_\mathrm{m}=0$ and $\gamma_{\mathrm{m}}=2$ correspond to a positive cosmological constant and a stiff fluid, respectively.
In turn, $T_{\mathrm{YM}\mu\nu}$ is written as
\begin{equation} 
\label{TmunuYM}
T_{\mathrm{YM} \mu\nu}=-\frac{1}{2e^2}\text{Tr}\left[F_{\mu\lambda}F^{~\lambda}_ {\nu}-\frac{1}{4}g_{\mu\nu}F_{\lambda\sigma}F^{\lambda\sigma}\right] -\frac{\mu^2}{2}\text{Tr}\left[2 A_\mu A_\nu- g_{\mu\nu}A_\lambda A^\lambda \right],
\end{equation}
where $e>0$ is the gauge coupling constant, $\mu\ge 0$ the mass of the gauge field, and the field strength 
\begin{equation} 
F_{\mu\nu}:=\partial_\mu A_\nu-\partial_\nu A_\mu+[A_\mu,A_\nu]
\end{equation} 
satisfies the Yang-Mills equation
\begin{equation}
\label{yang-mills-eq}
\nabla_\mu F_{\alpha\beta}+[A_\mu,F_{\alpha\beta}]=0.
\end{equation}
Following Ref.~\cite{bentoetal93}, we assume that the vector fields $\vec A$ have a global $SO(3)$ symmetry 
in which case $L_a$ are 
the infinitesimal generators of the internal $SO(3)$ group, satisfying $[L_a,L_b]=\varepsilon_{abc}L_c$. Imposing, furthermore, that $A_\mu$ is $E^3$ symmetric, and fixing the gauge freedom with the temporal (Hamiltonian) gauge leads to
\begin{equation}
\label{hamiltonian-gauge}
A_0=0,~~~~~~A^a_{~i}(t)= \chi (t) \delta^a_{~i},
\end{equation} 
where $\chi(t)$ is a $C^2$ function of $t$ and we denote space indices with $i=1,2,3$.   

It turns out that, under our symmetry assumptions, the Yang-Mills field has only a single "scalar" degree of freedom~\cite{bentoetal93,GV91} and, using 
$\text{Tr}[F_{\mu\nu}F^{\mu\nu}]=3(\dot \chi^2/a^2-\chi^4/(4a^4)),~\text{Tr}[A_{\mu}A^{\mu}]=-3\chi^2/(2a^2), ~\text{Tr}[F_{0i}F^{0i}]=-3\dot \chi^2/(2a^2)$, then the tensor $T_{\mathrm{YM}\mu\nu}$ can also be decomposed with respect to $\vec u$ on a "perfect fluid" form as in \eqref{TmunuPF}:
\begin{equation}
\label{TmunuYM2}
T_{{\mathrm{YM}}\mu\nu}=(\rho_{\mathrm{YM}}+p_{\mathrm{YM}}) u_{\mu} u_\nu+p_{\mathrm{YM}} g_{\mu\nu},
\end{equation}
for appropriate identifications of the quantities in \eqref{TmunuYM} with $\rho_{\mathrm{YM}}=T_{{\mathrm{YM}}00}$ and $p_{\mathrm{YM}}=(1/3)T_{{\mathrm{YM}}i}^{\quad\,\, i}$, 
which we give ahead in \eqref{rhoYM} using new scalar variables defined in~\eqref{Newscalars}. 
In particular the level sets of $\chi$ coincide with the surfaces of simultaneity of observers comoving with the fluid.

For the (conformal invariant) massless Yang-Mills field ($\mu=0$), the resulting stress-energy tensor is trace-free, so that its effective equation of state is that of a radiation fluid and the model is explicitly solvable~\cite{bentoetal93,GV91}. The massive case, $\mu\ne 0$, has been studied in~\cite{bentoetal93} using a dynamical systems approach, and the inclusion of a dust and radiation fluid has been discussed in~\cite{BYM05}.

The evolution and constraint equations are then obtained from \eqref{EFEs-general}, using \eqref{metric} on the left-hand-side (which gives $\text{Ric}$ and $S$), and using \eqref{Tmunugeral} on the right-hand-side satisfying \eqref{euler-eq} and \eqref{yang-mills-eq}, under the above assumptions. 
So, the Einstein-Euler-Yang-Mills system in a flat Robertson-Walker geometry reduces to the following system of nonlinear ODEs:
\begin{subequations}
	\begin{align}
	H^{2} &= \left(\frac{\dot{\chi}}{2\sqrt{2}a e}\right)^2+\left(\frac{\chi}{2^{\frac{1}{4}}2a\sqrt{e}}\right)^4 
	+\left(\frac{\mu\chi}{2a}\right)^2+\frac{\rho_\mathrm{m}}{3}.
	\label{Hsquared} \\
	\label{ddot-chi}
	\ddot{\chi} &= -H\dot{\chi}-\frac{\chi^{3}}{2a^2}-2 \mu^2 e^2\chi \\
	\dot{\rho}_\mathrm{m} &= -3H\gamma_\mathrm{m}\rho_\mathrm{m}\\
	\dot{H} &=  -\left(\frac{\dot{\chi}}{2ae}\right)^2 - \left(\frac{\chi}{2a\sqrt{e}}\right)^4-\left(\frac{\mu\chi}{2a}\right)^2 -\frac{\gamma_\mathrm{m}}{2}\rho_\mathrm{m} \\
		\dot{a} &=Ha \label{main1}
	\end{align}
\end{subequations}
where the overdot denotes a derivative with respect to $t$ and $H(t):=\dot{a}/a$ is the Hubble function.

Regarding $\dot{\chi}$ as a new dependent variable, the first equation can be seen as a constraint for the variables $(\chi,\dot{\chi},\rho_\mathrm{m},H,a)$. 
By further introducing
\begin{equation}\label{Newscalars}
\phi(t):=\frac{\chi}{\sqrt{2}a}~~~\text{and}~~~\psi(t):=\frac{\dot{\chi}}{\sqrt{2}ae},
\end{equation}
the equation for $a(t)$ decouples, and leaves a reduced dynamical system for the state vector $(\phi,\psi,\rho_\mathrm{m},H)$ given by
\begin{subequations}	\label{eqspsi1}
	\begin{align}
	\dot{\phi} &=-H\phi+e\psi \\
	\dot{\psi} &=-2H\psi -\frac{\phi^3}{e}-2\mu^2e\phi \\
	\dot{\rho}_\mathrm{m} &= -3H\gamma_\mathrm{m}\rho_\mathrm{m} \\
	\dot{H} &=  {-\frac{\psi^2}{2} - \frac{\phi^4}{4e^2} -\mu^2 \frac{\phi^2}{2}-\frac{\gamma_\mathrm{m}}{2}\rho_\mathrm{m}}, 
	\end{align}
\end{subequations}
with constraint
\begin{equation}
\label{eqspsi-last}
H^{2} = \frac{\psi^2}{4}+\frac{\phi^4}{8e^2}+\frac{\mu^2 }{2}\phi^2+\frac{\rho_\mathrm{m}}{3}.
\end{equation}
The Yang-Mills field generates an effective energy density $\rho_{\mathrm{YM}}\geq0$ and pressure $p_{\mathrm{YM}}$, given by
\begin{subequations}
\label{rhoYM}
	\begin{align}
	\rho_{\mathrm{YM}} (t)&:= 
	  3\Big[\frac{\psi^2}{4}+\frac{\phi^4}{8e^2}+\frac{\mu^2 }{2}\phi^2\Big] \\
	p_{\mathrm{YM}}(t) &:= 
	  \frac{\psi^2}{4}+\frac{\phi^4}{8e^2}-\frac{\mu^2 }{2}\phi^2,
	\end{align}
\end{subequations}
from which we define the function
\begin{equation}
\label{gammafunction}
\gamma_{\mathrm{YM}}(t) := 1+\frac{p_{\mathrm{YM}}}{\rho_{\mathrm{YM}}}\,. 
\end{equation}
In~\eqref{EFEs-general}, we have fixed physical units such that $8\pi G=c=1$, where $G$ is the Newton gravitational constant and $c$ the speed of light. With this choice, we have that  $[t]=L,~[H]=L^{-1},~[e]=L^{-1},~[\phi]=L^{-1},~[\psi]=L^{-1}$, whereas $\mu$ is dimensionless.

Our aim is to apply a new global dynamical systems formulation adapted from the problem of a minimally coupled scalar field having a zero local minimum of the potential, such as the Klein-Gordon field~\cite{alhugg15} or more general monomial potentials~\cite{alhetal15}. Similar methods have also been applied to $\alpha$-attractor $\mathrm{E}$ and $\mathrm{T}$-models of  inflation in~\cite{alhetal17} as well as to the Starobinskii model of modified $f(R)$ gravity theory~\cite{alhetal16}. 

The new formulation has several advantages with respect to the original variables and which are commonly used in the literature, see e.g.~Refs.\cite{bentoetal93} and \cite{bessaetal}. To see this, consider for simplicity the state vector $(\phi,\psi,H)$, i.e. with no fluid matter content. The state space consists of a surface defined by the constraint~\eqref{eqspsi-last} with the fixed point $\mathrm{M}$ located at $(0,0,0)$, which is the only fixed point of system~\eqref{eqspsi1} and\eqref{eqspsi-last}, see Fig.~\ref{SScone}. This fixed point joins the two disconnected parts, having either $H>0$ or $H<0$, i.e., preserving the sign of $H$. We are interested in expanding cosmologies so, from now on, we will restrict the analysis to the upper half of the state space where $H>0$. By solving for $H$ in~Ref.\eqref{eqspsi-last} and inserting the positive root in the evolution equations, leads to an unconstrained two-dimensional dynamical system on the plane. This system might have differentiability problems at the origin, where lies the full degenerated Minkowski fixed point $\mathrm{M}$. The blow up of such fixed point can be found in ~\cite{bentoetal93} where it was shown that it is a local focus. 

However, as it will be shown here, in the present formulation this fixed point appears naturally as a periodic orbit and provides indeed the correct picture
(this also clarifies the issue of asymptotic self-similarity and manifest self-similarity breaking as discussed in~\cite{alhetal15}). This fact is related to the existence of a conserved quantity for the system: the expansion normalized effective energy density due to the Yang Mills field 
\begin{equation}
\Omega_{\mathrm{YM}}:=\frac{\rho_{\mathrm{YM}}}{3H^2}.
\end{equation} 

Another relevant aspect of this formulation concerns the compactification of the state space on the plane, in which Poincar\'e method is usually the standard approach. The Poincar\'e compactification does not take into account the natural topological state-space structures inherent to each particular model, and might lead to expensive  computations (as exemplified in~Ref.\cite{bessaetal}). 
Instead, the use of (dimensionless) expansion normalized variables, gives a very natural compactification of the state space  (where $H\rightarrow+\infty$), in which self-similar solutions appear as hyperbolic fixed points.

Furthermore, when introducing matter in the form of a perfect fluid with a linear equation of state, the state space becomes the region limited by a quadratic surface (see Fig. \ref{SScone}), and the old formulation would also lead to difficulties when discussing asymptotic source dominance since all orbits tend to a single degenerated fixed point $\mathrm{M}$. 
Instead, the correct picture of attractors being periodic orbits leads naturally to the use of averaging techniques from dynamical systems theory, see e.g.~\cite{SVM00}. This, in turn, allows us to give rigorous proofs concerning the asymptotics when matter models other than the Yang-Mills field are present. 

Finally, this framework is the starting point for considering less restrictive geometries like in the spatially homogeneous but anisotropic spacetimes, an issue we shall discuss further in Sec.\ref{CR}
\begin{figure}[ht!]
	\begin{center}
		\includegraphics[width=0.30\textwidth]{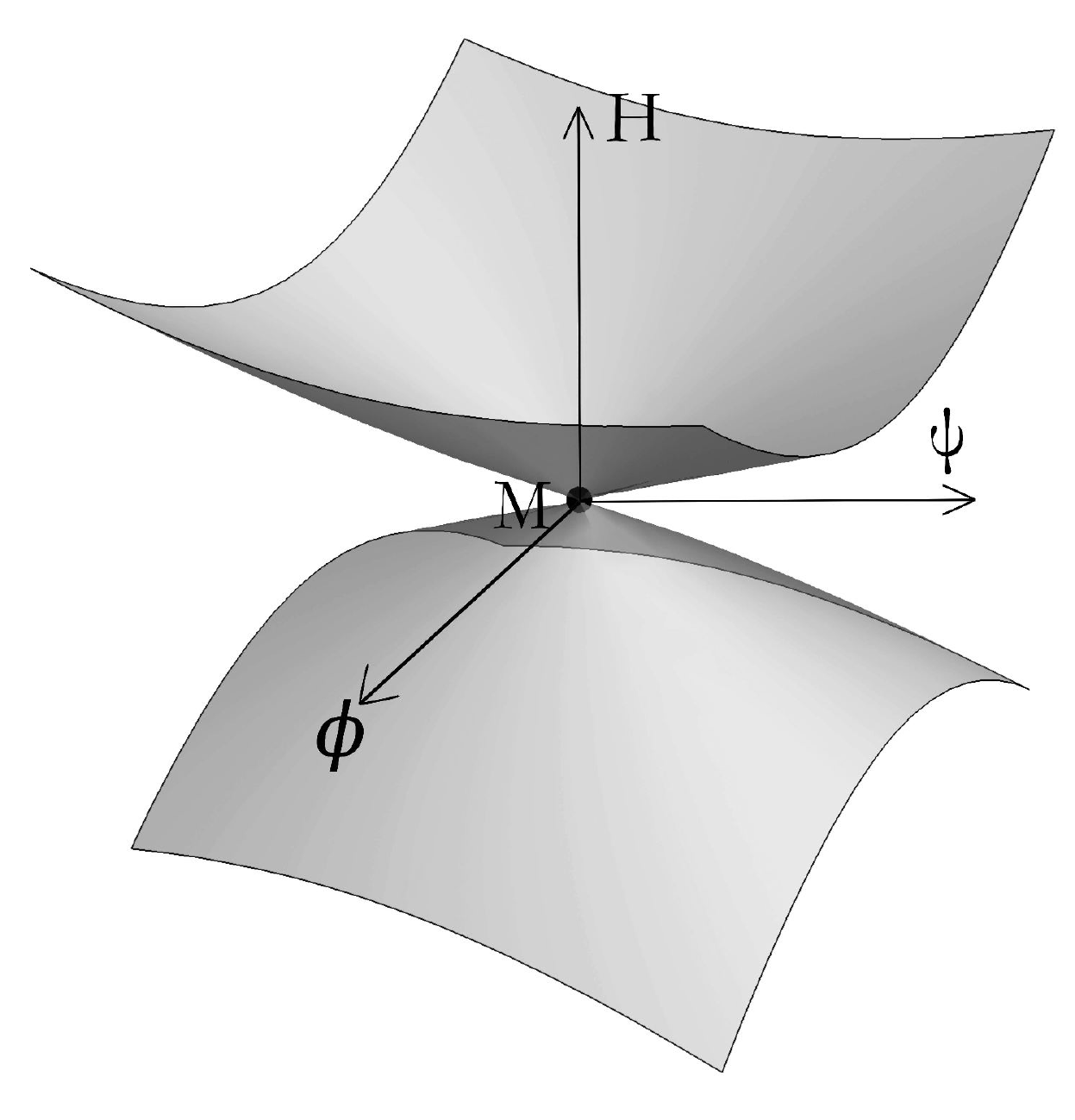}
	\end{center}
	\vspace{-0.5cm}
	\caption{The state-space of system~\eqref{eqspsi1} defined by the constraint~\eqref{eqspsi-last}.}
	\label{SScone}
\end{figure}

The paper is structured as follows: in Section~\ref{massless}, we consider the simplest model of a massless Yang-Mills field and a fluid with linear equation of state. We reformulate the Einstein field equations to a 3-dimensional dynamical system on a compact state-space, followed by an analysis of the flow which yields a global description of the solution space including its asymptotic behavior. For the pure massless Yang-Mills invariant subset ($\rho_\mathrm{m}=0$), the field equations can be further reduced to an analytical 2-dimensional unconstrained dynamical system which is integrable in terms of elliptic functions, thus contextualizing this well-known explicit  solutions in a global (compact) state-space picture.
In Section~\ref{massive}, we consider the massive Yang-Mills field together with the fluid matter model. In this case, the field equations are reformulated as a 4-dimensional dynamical system with a constraint (for a study of constraint systems in cosmology see Refs.\cite{WainElis} and \cite{HW92} ). We make a global analysis of the flow and give rigorous proofs concerning the asymptotic behavior of general solutions both in the past and future time directions. 
 We conclude the paper with a brief discussion of Yang-Mills fields in anisotropic cosmologies.
\section{Massless Yang-Mills field (case $\mu=0$)}
\label{massless}
For the massless Yang-Mills field, the ratio $p_{\mathrm{YM}}/\rho_{\mathrm{YM}}$ is constant and the function $\gamma_{\mathrm{YM}}(t)$, defined in \eqref{gammafunction}, is simply given by
\begin{equation}
\gamma_{\mathrm{YM}}= \frac{4}{3}.
\end{equation}
Hence, the massless Yang-Mills field can be view as an effective radiation fluid, which basically turns the problem into that of a two-fluid cosmology. However, it is instructive to consider first this simple model, since it allows us to introduce some basic definitions and illustrate how a global dynamical systems formulation of the original equations can be constructed. It will also allow us to situate well-known explicit solutions in a global state space picture, as well as emphasizing the differences that arise in the more complicated case of the massive Yang-Mills field. 

We assume an expanding cosmology $H(t) > 0$, and introduce the (dimensionless) $H$-normalized variables
\begin{equation}
X_1=\frac{\phi}{2^{\frac{3}{4}}\sqrt{eH}},\qquad\Sigma_{\mathrm{YM}}=\frac{\psi}{2H},\qquad  \tilde{T}=\sqrt{\frac{\sqrt{2}e}{H}}, \qquad \Omega_{\mathrm{m}}=\frac{\rho_{\mathrm{m}}}{3H^2},
\end{equation}
together with the number of $e-$folds $N=\ln{(a/a_0)}$, where $a_0$ is some reference epoch at which $N=0$, and
\begin{equation}\label{e-folds}
\frac{dN}{dt}= H.
\end{equation}
Then, the system of Eqs. \eqref{eqspsi1} and \eqref{eqspsi-last}, in the new variables, reduces to a \emph{local} 3-dimensional dynamical system
 \begin{subequations}\label{Dyn1}
  \begin{align}
  \frac{dX_1}{dN} &= -\frac{1}{2}\left[(1-q)X_1-2\tilde{T}\Sigma_{\mathrm{YM}}\right]\\
   \frac{d\Sigma_{\mathrm{YM}}}{dN} &=-\left[(1-q)\Sigma_{\mathrm{YM}}+2\tilde{T}X^3_1\right] \\
   \frac{d\tilde{T}}{dN} &=\frac{1}{2}(1+q)\tilde{T}, 
  \end{align}
 \end{subequations}
where we make use of the fact that the constraint 
\begin{equation}
\label{constr}
1-\Omega_{\text m}=X^4_1+\Sigma^{2}_{\mathrm{YM}} \,
\end{equation}
is linear in $\Omega_\mathrm{m}$, to solve for $\Omega_{\mathrm{m}}$, and where we introduced  the so-called {\em deceleration parameter} $q$, defined via $\dot{H}=-(1+q)H^2$, i.e.
\begin{equation}
\begin{split}
\label{q-equation}
q 
&=-1+2\left(\Sigma^{2}_{\mathrm{YM}}+X^4_1\right)+\frac{3}{2}\gamma_{\text m}\Omega_{\text m}  \\
&=  1+\frac{3}{2}\left(\gamma_{\text m}-\frac{4}{3}\right)\left(1-\Sigma^{2}_{\mathrm{YM}}-X^4_1\right).
\end{split}
\end{equation}
Since $\Omega_\mathrm{m}\geq0$, the constraint equation~\eqref{constr} implies that
\begin{equation}\label{BoundVar}
 -1\leq X_1\leq 1 , \qquad -1\leq\Sigma_{\mathrm{YM}}\leq 1 ,\qquad 0\leq\Omega_\mathrm{m}\leq 1. 
 \end{equation}
Moreover, since $0<\gamma_{\mathrm{m}}\leq2$, it follows from \eqref{q-equation} that 
\begin{equation}
-1<q\leq 2\,.
\end{equation}
Hence, the right-hand side of \eqref{Dyn1} becomes unbounded only when $\tilde T\rightarrow+\infty$ ($H\rightarrow 0$). In order to obtain a global dynamical systems formulation on a compact state space, we further introduce
\begin{equation}\label{Ttilde}
T=\frac{\tilde{T}}{1+\tilde{T}}\,
\end{equation}
so that $T\rightarrow 0$ as $\tilde{T}\rightarrow 0$, and $T\rightarrow 1$ as $\tilde{T}\rightarrow +\infty$. 
We also introduce a new independent variable $\tau$ defined by
\begin{equation}
\frac{d\tau}{dt}=\frac{H}{1-T}.
\end{equation}
The $\tau$ variable is constructed such that it  interpolates between the two asymptotic regimes described by the different scales inherent to the model, i.e. the Hubble scale $H$, when $H\rightarrow+\infty$, and the scale associated with gauge-coupling constant $e$, when $H\rightarrow0$; see Refs.\cite{alhugg15} and \cite{alhetal15} for more details on this issue.  
This leads to a \emph{global} 3-dimensional dynamical system,
\begin{subequations}\label{DynSys2}
 \begin{align}
 \frac{dX_1}{d\tau} &= -\frac{1}{2}\left[(1-q)(1-T)X_1-2T\Sigma_{\mathrm{YM}}\right] \label{dyn1}\\
   \frac{d\Sigma_{\mathrm{YM}}}{d\tau} &=-\left[(1-q)(1-T)\Sigma_{\mathrm{YM}}+2TX^3_1\right] \label{dyn3}\\
  \frac{dT}{d\tau} &=\frac{1}{2}(1+q) T (1-T)^2,
 \end{align}
\end{subequations}
where the constraint~\eqref{constr} is used to globally solve for $\Omega_{\mathrm{m}}$ and $q$ is given by~\eqref{q-equation}. It is also useful to consider the auxiliary evolution equation for $\Omega_{\mathrm{YM}}:=\rho_\mathrm{YM}/(3H^2)=\Sigma^{2}_{\mathrm{YM}}+X^4_1$ (equivalently $\Omega_{\mathrm{m}}=1-\Omega_{\mathrm{YM}}$), which is given by 
\begin{equation}\label{OmegaYMEvEq}
\frac{d\Omega_{\mathrm{YM}}}{d\tau} = 3(1-T) (\gamma_{\mathrm{m}} -\frac{4}{3})\Omega_{\mathrm{YM}}(1-\Omega_{\mathrm{YM}}).
\end{equation}
The state space ${\bf S}$ is a 3-dimensional space consisting of a deformed solid cylinder of height $0<T<1$. The outer shell of the cylinder corresponds to the \emph{pure Yang-Mills} invariant  subset $\Omega_{\text m}=0$ ($\Omega_{\mathrm{YM}}=1$) which we denote by ${\bf S}_{\mathrm{YM}}$. The axis of the cylinder is a straight line with $\Omega_{\text m}=1$ ($\Omega_{\mathrm{YM}}=0$) and corresponds to the invariant subset associated with the (self-similar) flat Friedmann-Lema\^itre (FL) spacetime. 
The state space ${\bf S}$ can be analytically extended to include its closure, i.e., the invariant boundaries $T=0$ and $T=1$, and form the extended state space $\overline{{\bf S}}$, while the extension of ${\bf S}_{\mathrm{YM}}$ to $T=0$ and $T=1$ will be denoted by ${\bf \overline S}_{\mathrm{YM}}$. This extension is crucial since all attracting sets are located on these boundaries as shown by the following simple lemma:
\begin{lemma}
	\label{lemma1}
	The $\alpha$-limit set of all interior orbits in $\mathbf{S}$ is located at $T=0$, while the $\omega$-limit set of all interior orbits in $\mathbf{S}$ is located at $T=1$.
\end{lemma}
\begin{proof}
Since $1+q>0$, then $T$ is strictly monotonically increasing in the interval $(0,1)$. By the monotonicity principle, it follows that there are no fixed points, recurrent or periodic orbits in the interior of the state space ${\bf S}$, and the $\alpha$ and $\omega$-limit sets of all orbits in ${\bf S}$ are contained at $T=0$ and $T=1$, respectively.
\end{proof}
We now give a detailed description of the invariant boundaries $T=0$, associated with the asymptotic past ($H\rightarrow+\infty$), and $T=1$, associated with the asymptotic future ($H\rightarrow 0$), as well as the pure massless Yang-Mills invariant subset $\mathbf{S}_{\mathrm{YM}}$ and the Friedmann-Lema\^itre invariant subset. 
\subsection*{A. The invariant boundary $T=0$}\label{T0}
The flow induced on the $T=0$ boundary is given by
\begin{subequations}
	\begin{align}
	\frac{dX_1}{d\tau} &= \frac{3}{4}\left(\gamma_{\mathrm{m}}-\frac{4}{3}\right)\Omega_{\text m} X_1,  \\
		\frac{d\Sigma_{\mathrm{YM}}}{d\tau} &=\frac{3}{2}\left(\gamma_{\text m}-\frac{4}{3}\right)\Omega_{\text m} \Sigma_{\mathrm{YM}} 	
	\end{align}
\end{subequations}
with the constraint $\Omega_{\text m} = 1-\Omega_{\mathrm{YM}} = 1- X^4_1-\Sigma^2_{\mathrm{YM}}$.
For $\gamma_{\mathrm{m}}=4/3$, this subset consists only of fixed points, forming the deformed disk:
\begin{equation*}
\mathrm{D_R}:\qquad 0\leq\Sigma^2_{\mathrm{YM}}+X^4_1\leq 1,\qquad {\text{for}}\qquad T=0.
\end{equation*}
For $\gamma_{\mathrm{m}}\neq 4/3$, the invariant subset $\Omega_{\mathrm{m}}=0$ consists of a deformed circle of fixed points given by
\begin{equation*}
\text{L}_{\text{R}}:\qquad \Sigma^2_{\mathrm{YM}}+X^4_1=1,\qquad {\text{for}}\qquad T=0, 
\end{equation*}
and there is one more isolated fixed point $\mathrm{FL_0}$ located at $\Omega_{\mathrm{YM}}=0$, i.e. 
\begin{equation*}
{\text{FL}}_0:\quad\Sigma_{\mathrm{YM}} = X_1 = 0,\qquad {\text{for}}\qquad T=0. 
\end{equation*}
At the invariant boundary $T=0$, the trajectories of the solutions are easily found by quadrature giving
\begin{equation}
\Sigma_{\mathrm{YM}} = C X^2_1,
\end{equation}
where $C$ is a real constant that parametrizes the solutions. This equation clearly shows that the flow is invariant under the transformation $(X_1,\Sigma_{\mathrm{YM}})\rightarrow (-X_1,\Sigma_{\mathrm{YM}})$. Moreover, since $\Omega_{\text m}>0$, a straightforward inspection of the flow, shows that, if $\gamma_{\text m}>\frac{4}{3}$ (resp. $\gamma_{\text m}<\frac{4}{3}$), then $\mathrm{L_R}$ is a  sink (resp. source) of a 1-parameter set of solutions with a single solution ending (resp. originating) from each fixed point and $\mathrm{FL_0}$ is a source (resp. sink) of a 1-parameter set of solutions, see Fig.~\ref{fig:T0_Massless}.
\begin{figure}[ht!]
	\begin{center}
			\subfigure[Invariant boundary $T=0$ for $\gamma_{\text m}=1<4/3$.]{\label{fig:Dust}
			\includegraphics[width=0.30\textwidth]{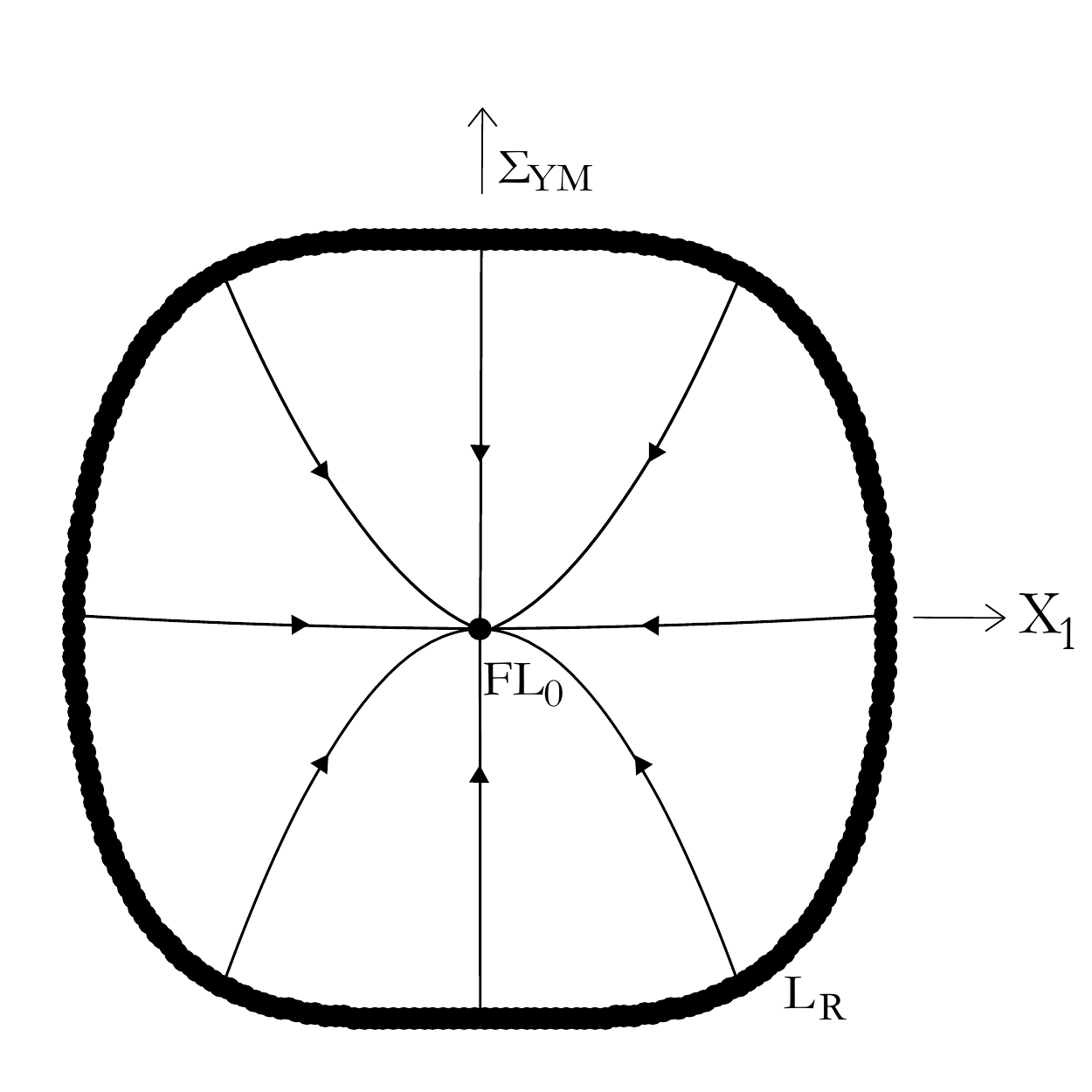}} \hspace{1cm}
		\subfigure[Invariant boundary $T=0$  for $\gamma_{\text m}=\frac{3}{2}>4/3$.]{\label{fig:Stiffer}
			\includegraphics[width=0.30\textwidth]{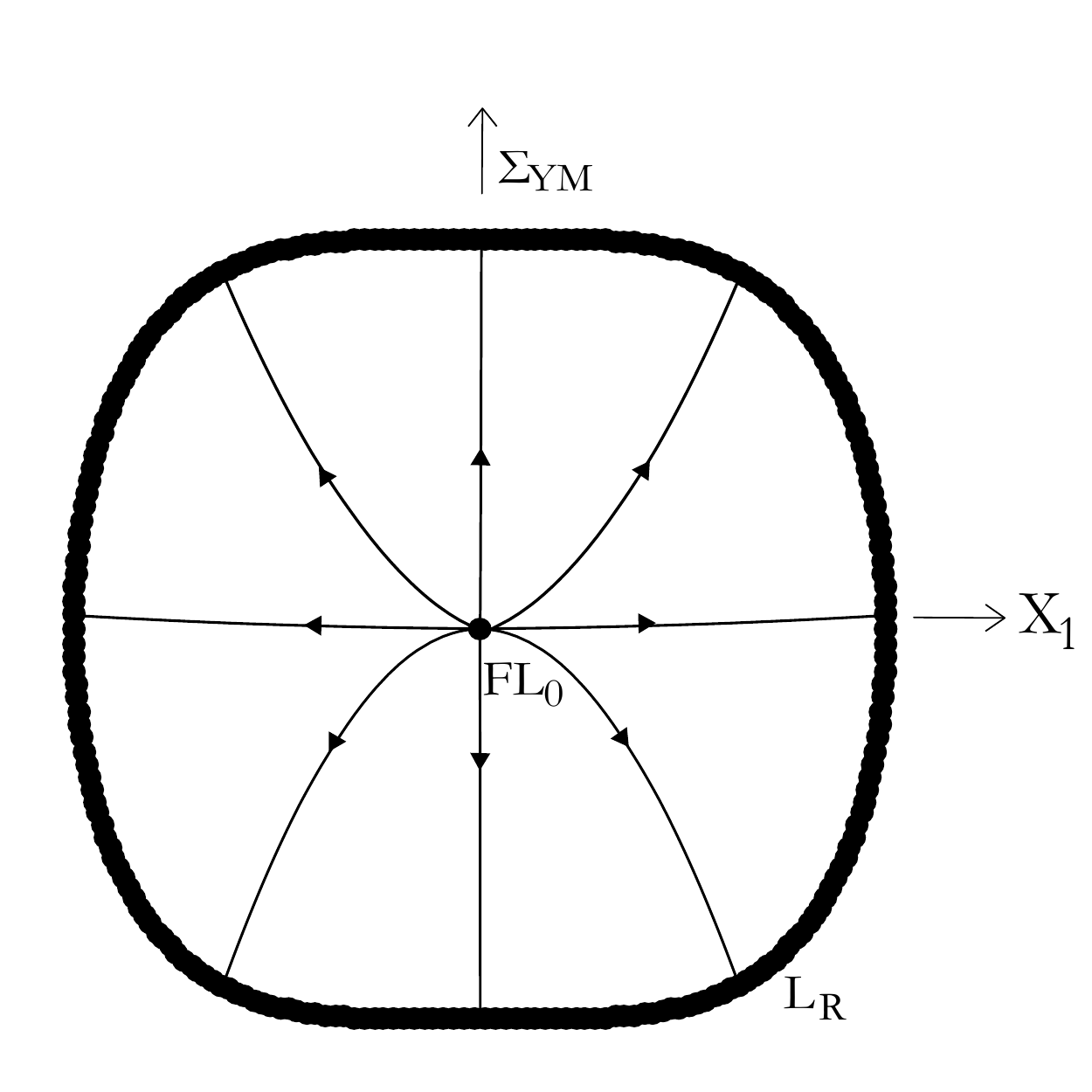}}
	\end{center}
	\vspace{-0.5cm}
	\caption{The invariant boundary $T=0$ of phase-space ${\bf S}$ for two different values of $\gamma_{\text m}$. The picture for $\gamma_{\text m}=4/3$ consists of a disk of fixed points.}
	\label{fig:T0_Massless}
\end{figure}
%
\subsection*{B. The invariant boundary $T=1$}
\label{T1}
On the $T=1$ invariant boundary, the system \eqref{dyn1} and \eqref{dyn3} reduces to
	\begin{subequations}
	\begin{align}
	\frac{dX_1}{d\tau} &= \Sigma_{\mathrm{YM}}, \\
	\frac{d\Sigma_{\mathrm{YM}}}{d\tau} &=-2X^3_1 ,
	\end{align}
\end{subequations}
which has a single fixed point: 
\begin{equation*}
{\text{FL}}_1:\quad\Sigma_{\mathrm{YM}} = X_1 = 0,\qquad {\text{for}}\qquad T=1. 
\end{equation*}
In this case, it also follows that $d\Omega_{\mathrm{YM}}/d\tau = 0$, implying
\begin{equation}
\Omega_{\mathrm{YM}}=C,
\end{equation}
where $C\in [0,1]$. The $T=1$ boundary is then foliated by a 1-parameter set of periodic orbits $\mathcal{P}_{\Omega_{\mathrm{YM}}}$ and, therefore, the fixed point $\mathrm{FL_1}$ (corresponding to $C=0$) is a center [see Figure~\ref{fig:T1Massless}]. Note that $C=1$ gives the outer periodic orbit $\mathcal{P}_{1}$ with $\Omega_{\mathrm{YM}}=1$ ($\Omega_{\text m}=0$).
\subsection*{C. The Friedmann-Lema\^itre invariant subset: $\mathrm{FL_0} \rightarrow \mathrm{FL_1}$}
The invariant subset $\Omega_{\text m}=1$  consists of a straight   heteroclinic orbit connecting the $\mathrm{FL}_0$ fixed point, located at the origin, to the 
fixed point $\mathrm{FL}_1$ located at $(X_1,\Sigma_{\mathrm{YM}},T)=(0,0,1)$. This orbit is associated with the flat Friedmann-Lema\^itre solution, where $T$ describes the evolution of $H$; see Figure~\ref{fig:Omega1}.
\begin{figure}[ht!]
	\begin{center}
		\subfigure[Invariant boundary $T=1$.]{\label{fig:T1Massless}
			\includegraphics[width=0.30\textwidth]{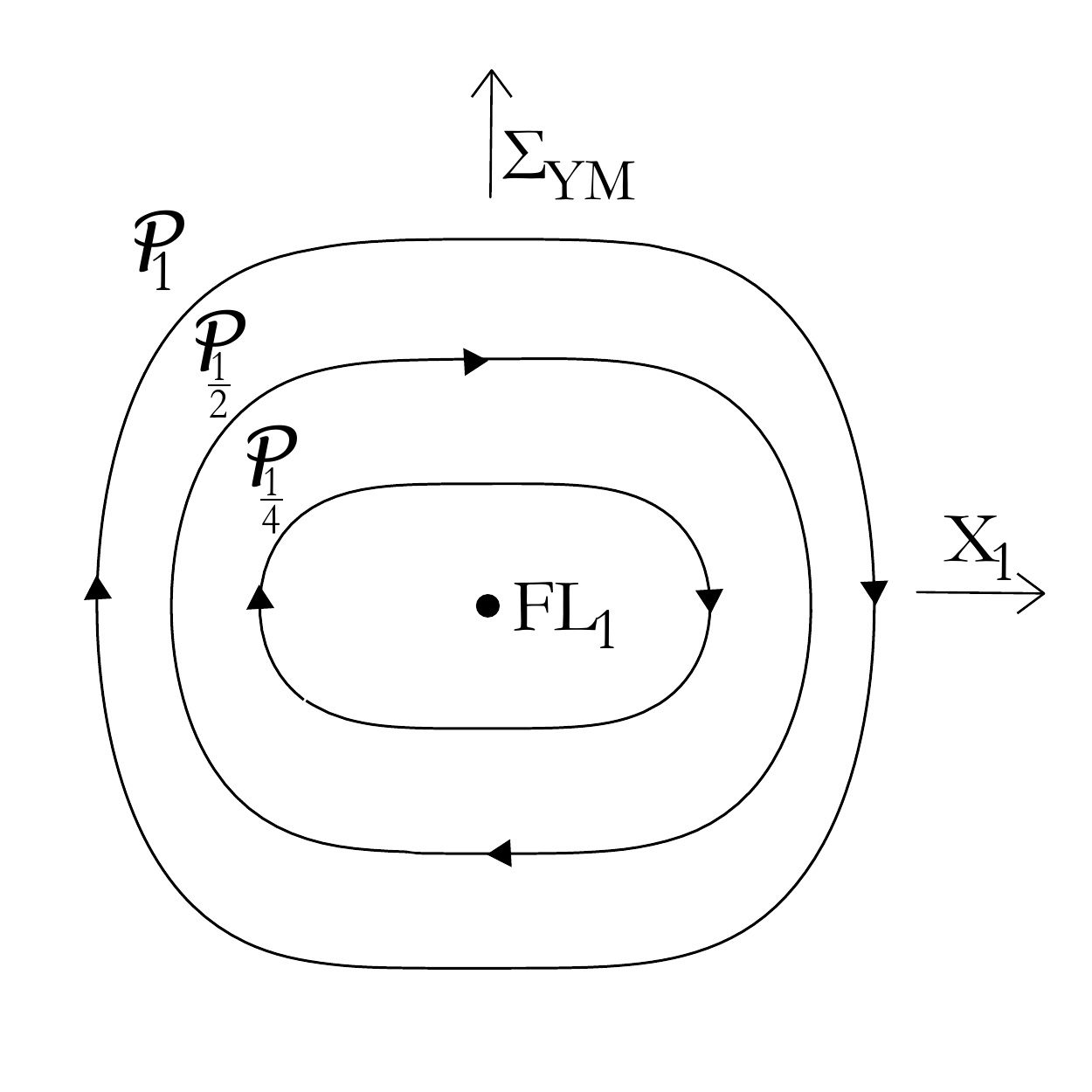}}
		\hspace{1cm}
		\subfigure[Invariant subset $\Omega_m=1$.]{\label{fig:Omega1}
			\includegraphics[width=0.30\textwidth]{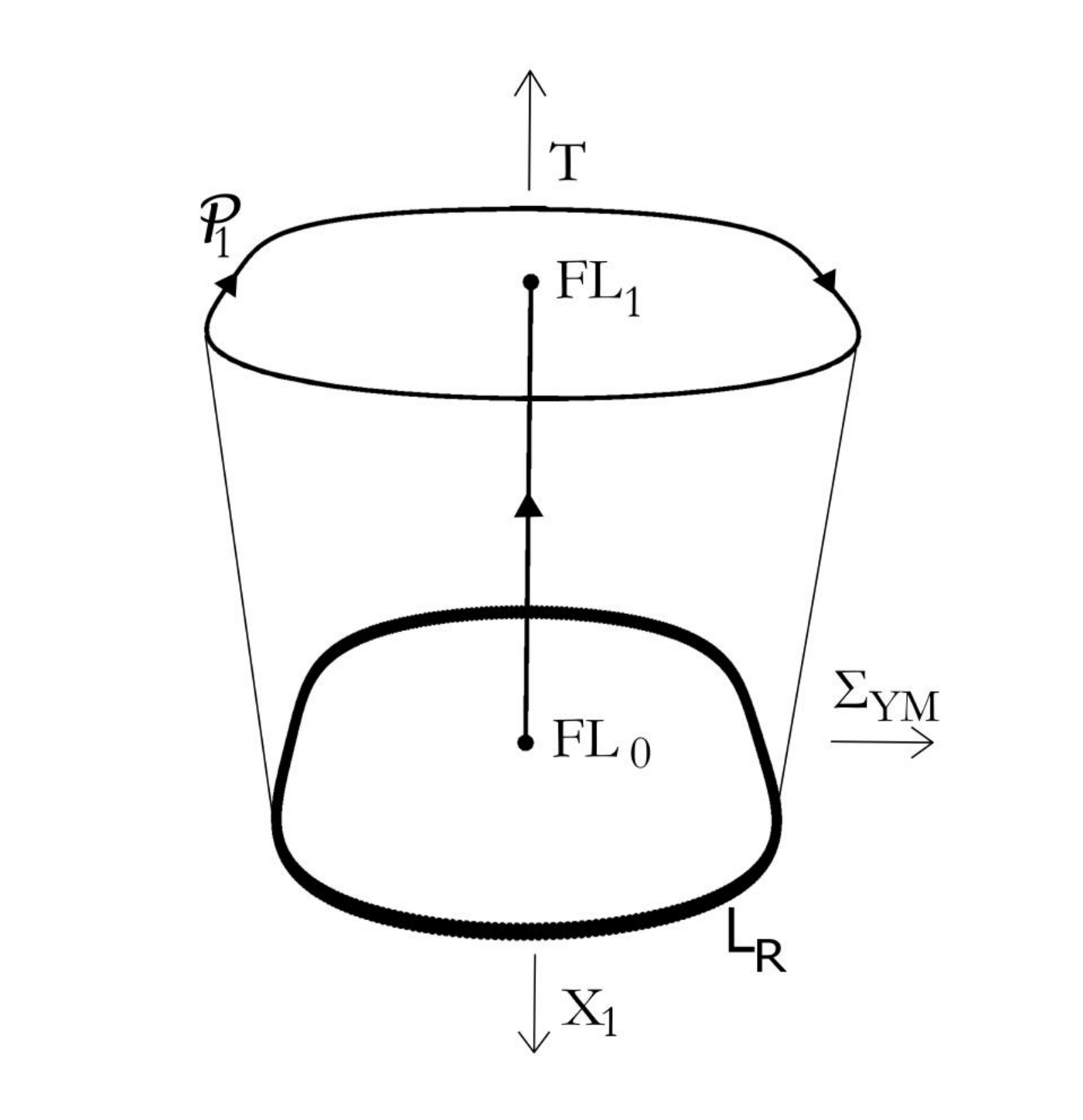}}
		\end{center}
		\vspace{-0.5cm}
		\caption{Representation of the invariant boundary $T=1$ and of the invariant subset $\Omega_{\text m}=1$.}
	\end{figure}
\subsection*{D. The pure massless Yang-Mills subset $\overline{{\bf S}}_{\mathrm{YM}}$}\label{TF}
On the invariant set $\Omega_{\text m}=0$, it follows that the deceleration parameter $q$ is constant, with $q=1$, and the dynamical system simplifies to
\begin{equation}
\frac{dX_1}{d\tau} = T\Sigma_{\mathrm{YM}} \,, \quad
  \frac{d\Sigma_{\mathrm{YM}}}{d\tau} =-2TX_1^3 \,, \quad 
\frac{dT}{d\tau} = T (1-T)^2, 
\end{equation}
subject to the constraint
\begin{equation}
 \Sigma^{2}_{\mathrm{YM}}+X_1^4=1.
\end{equation}
This constraint can be globally solved by introducing the angular variable $\theta$ as
\begin{equation}
X_1=\cos{\theta},\qquad  \Sigma_{\mathrm{YM}}=G(\theta)\sin{\theta}, 
\end{equation}
where
\begin{equation}
 G(\theta)=
 \sqrt{1+\cos^{2}{\theta}}.
\end{equation}
This leads to a 2-dimensional unconstrained dynamical system for the state vector $(\theta,T)$, given by
\begin{subequations}
\begin{align}
\label{theta1}
 \frac{d\theta}{d\tau} &=- TG(\theta)\\
 \label{theta2}
 \frac{dT}{d\tau} &=T(1-T)^2.
\end{align}
\end{subequations}
The intersection with the invariant boundary $T=0$, consists of the circle of fixed points $\mathrm{L}_{\mathrm{R}}$ whose linearisation yields  the eigenvalues $1$ and $0$, with the center manifold being the line itself, i.e., the circle of fixed points 
is normally hyperbolic, so that a unique solution originates from each fixed point $(\theta_0,0)$, $\theta_0\in[0,2\pi)$, and a one-parameter set of solutions (parameterized by $\theta_0$) originates from the circle into the interior of the state space $\mathbf{S}_{\mathrm{YM}}$. At $T=1$, it follows that
\begin{equation}
\frac{d\theta}{d\tau} = -G(\theta)<0,
\end{equation}
which corresponds to the periodic orbit $\mathcal{P}_1$. From the monotonicity of $T$, see Lemma~\ref{lemma1}, it follows that all solutions originate from the circle of fixed points at 
$T=0$ and end at the periodic orbit at $T=1$ which, therefore, constitutes a limit cycle. 
In fact, using \eqref{theta1} and \eqref{theta2}, we find that, in this case, the orbits are the solutions to the equation
\begin{equation}
 \frac{d\theta}{dT}=-\frac{G(\theta)}{(1-T)^2},
\end{equation}
and which are given by
\begin{equation}
\label{theta-eq}
\theta(T)=F\left(\sqrt{2}\left(\frac{1}{1-T_0}-\frac{1}{1-T}\right)\Big |\frac{1}{\sqrt{2}}\right),
\end{equation}
where $F(x|k)$ is the Jacobi elliptic amplitude, satisfying $F(0|k)=0$. This 1-parameter set of solutions parametrized by $T_0$,  corresponds to the well-known solutions for the pure massless Yang-Mills field in a flat Robertson-Walker geometry found in Ref.\cite{bentoetal93} and \cite{GV91} by solving $d^2\chi/d\eta^2=-\chi^3/2$, where $\eta$ is the conformal time  $d\eta=dt/a(t)$. These solutions are depicted in Fig. \ref{fig:MasslessSpace} for different initial conditions.
\begin{figure}[h!]
\begin{center}
\subfigure[Dynamics on the invariant boundary $\overline{{\bf S}}_{\mathrm{YM}}$.]{\label{fig:Massless_Cyl}
\includegraphics[width=0.30\textwidth]{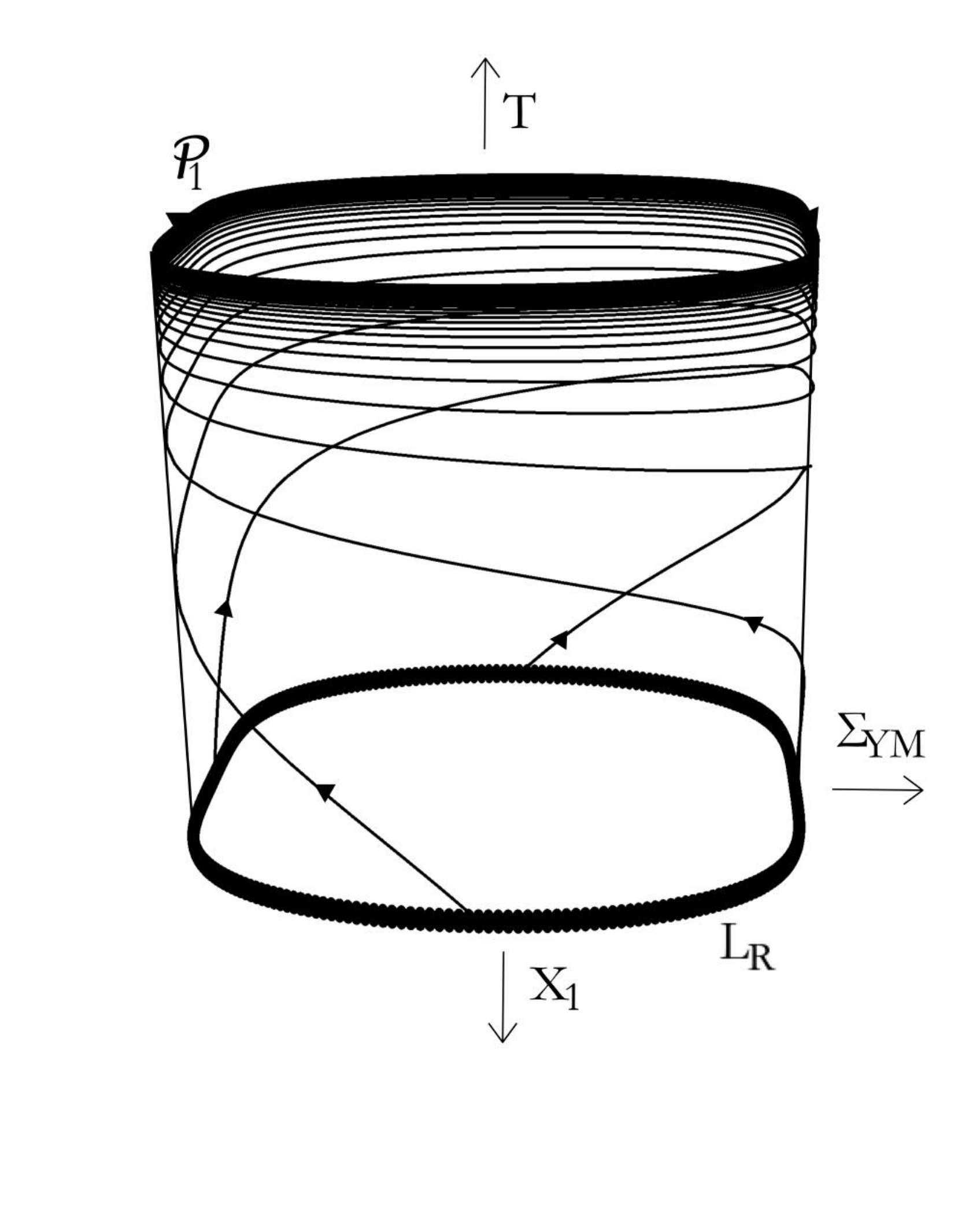}}
\hspace{1cm}
\subfigure[`Unwrapped' solution space, corresponding to solutions \eqref{theta-eq} for different values of $\theta_0$. ]{\label{fig:MasslessUnrapped}
\includegraphics[width=0.30\textwidth]{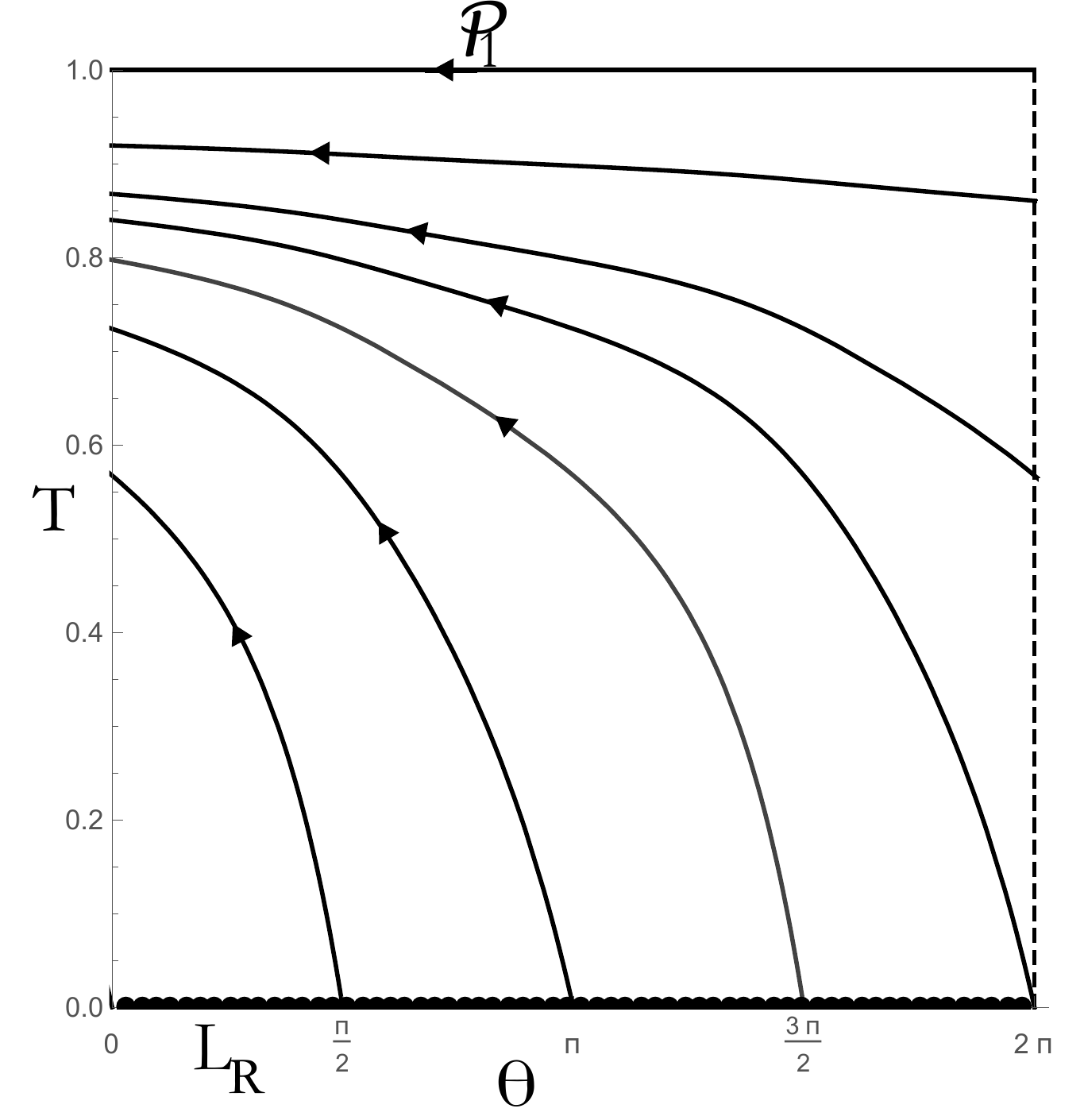}}
\end{center}
\vspace{-0.5cm}
\caption{Dynamics on the invariant set $\overline{{\bf S}}_{\mathrm{YM}}$.}
\label{fig:MasslessSpace}
\end{figure}
\subsection*{E. Global dynamics for massless Yang-Mills field and perfect fluid}
We now make use of the previous analysis to prove the following result:
\begin{proposition} Consider solutions of the system \eqref{DynSys2} with $0<\Omega_{\mathrm m}<1$:
\label{prop1}
\begin{itemize}
  \item[(i)] If $\gamma_{\mathrm m} > \frac{4}{3}$, then all solutions  converge, for $\tau\rightarrow-\infty$, to the fixed point $\mathrm{FL_0}$ with $\Omega_\mathrm{m}=1$ and, for $\tau\rightarrow+\infty$, to the outer periodic orbit $\mathcal{P}_1$ with $\Omega_{\mathrm m}=0$.
  \item[(ii)] If $0<\gamma_{\mathrm m} < \frac{4}{3}$, a 1-parameter set of solutions converges, for $\tau\rightarrow-\infty$, to each point on the circle of fixed point $\mathrm{L_R}$ with $\Omega_{\mathrm{m}}=0$, while all solutions converge, for $\tau\rightarrow+\infty$, to the fixed point $\mathrm{FL_1}$ with $\Omega_{\mathrm{m}}=1$.
  \item[(iii)] If $\gamma_{\mathrm m}= \frac{4}{3}$, a unique solution converges, for $\tau\rightarrow-\infty$, to each point on  the disk of fixed points $\mathrm{D_R}$, while a 1-parameter set of solutions converges, for $\tau\rightarrow+\infty$, to each inner periodic orbit $\mathcal{P}_{\Omega_{\mathrm{YM}}}$.
\end{itemize}
\end{proposition}
This means that in case $\gamma_{\mathrm{m}}>\frac{4}{3}$ (resp. $\gamma_{\mathrm{m}}<\frac{4}{3}$), the model is past (resp. future) asymptotic fluid dominated and future (resp. past) asymptotic Yang-Mills field dominated. In the critical case, $\gamma_{\mathrm{m}}=\frac{4}{3}$, the model in neither fluid nor Yang-Mills dominated towards the asymptotic past nor the asymptotic future, see Figure~\ref{fig:MasslessMatterSpace} for representative solutions.

\begin{proof}
	The proof makes use of Lemma~\ref{lemma1} and the simple orbit structure on the invariant boundaries, given in the previous subsections \ref{T0}A-\ref{TF}D , which imply that the only possible $\alpha$-limit sets are fixed points on $T=0$, while the $\omega$-limit sets can be either periodic orbits or the fixed point $\mathrm{FL}_{1}$ on $T=1$. 
	
	In order to prove the general asymptotic behavior, we make use of the auxiliary Eq.~\eqref{OmegaYMEvEq} for $\Omega_{\mathrm{YM}}$. Since $\gamma_{\mathrm{YM}}=4/3$ is constant, then Eq. \eqref{OmegaYMEvEq}, together with the evolution equation for $T$, can be easily solved for $\Omega_{\mathrm{YM}}$ in terms of $T$. For solutions with $0<\Omega_{\mathrm{YM}}<1$, and $\gamma_{\mathrm{m} }\ne \frac{4}{3}$, we get
	\begin{equation*}
	\left(\frac{\Omega^{\frac{3\gamma_{\text m}}{4}}_{\mathrm{YM}}}{1-\Omega_{\mathrm{YM}}}\right)^{\frac{1}{3\gamma_{\text m}-4}}=C\frac{T}{1-T},
	\end{equation*}
	where $C> 0$ is a real constant parametrizing the solutions. 
	The last equation clearly shows that if $\gamma_{\text m} >\frac{4}{3}$, then $\Omega_{\mathrm{YM}}\rightarrow 0$ as $T\rightarrow0$, and $\Omega_{\mathrm{YM}}\rightarrow 1$ as $T\rightarrow1$, i.e. all solutions with $0<\Omega_{\mathrm{YM}}<1$ start at $\mathrm{FL}_0$ and end at $\mathcal{P}_1$. In turn, if $\gamma_{\text m} <\frac{4}{3}$, then $\Omega_{\mathrm{YM}}\rightarrow 1$ as $T\rightarrow0$, and  $\Omega_{\mathrm{YM}}\rightarrow 0$ as $T\rightarrow 1$, i.e. all solutions start at $\mathrm{L}_\mathrm{R}$ and end at $\mathrm{FL}_1$. If $\gamma_{\text m}=\frac{4}{3}$, then $\Omega_{\mathrm{YM}}=C$, with $C\in (0,1)$ for all $T$, i.e. the solutions start at $\mathrm{D}_{\mathrm{R}}$ and end at $\mathcal{P}_{\Omega_{\mathrm{YM}}}$. 
	
	Now, we give a more precise description of the flow near the invariant boundaries $T=0$ and $T=1$. The linearisation of the system~\eqref{DynSys2} around the fixed points located at $T=0$ yields:
	\begin{itemize}
		\item {$\mathrm{FL}_0$}: eigenvalues   $\frac{3}{4}\left(\gamma_{\mathrm{m}}-\frac{4}{3}\right)$, $\frac{3}{2}\left(\gamma_{\mathrm{m}}-\frac{4}{3}\right)$ and $\frac{3}{4}\gamma_{\mathrm{m}}$, with associated eigenvectors $(1,0,0)$, $(0,1,0)$ and $(0,0,1)$. 
	\end{itemize}
	\begin{itemize}
		\item {$\mathrm{L}_\mathrm{R}$}: eigenvalues $0$, $-3\left(\gamma_{\mathrm{m}}-\frac{4}{3}\right)$ and $1$, with associated eigenvectors $(\Sigma_{\mathrm{YM}},-2X_{1}^3,0)$, $(X_1,2\Sigma_{\mathrm{YM}},0)$ and $(0,0,1)$ where $\Sigma^2_{\mathrm{YM}}+X^4_1=1$.
	\end{itemize}
	\begin{itemize}
		\item {$\mathrm{D}_\mathrm{R}$}: eigenvalues $0$, $0$ and $1$, with eigenvectors $(1,0,0)$, $(0,1,0)$ and $(0,0,1)$. 
	\end{itemize}
	For all interior orbits in $\bar{\mathbf{S}}$: When $\gamma_{\mathrm{m}}>4/3$, $\mathrm{FL}_0$ is a source of a 2-parameter set of orbits and, from  $\mathrm{L}_\mathrm{R}$, originates a 1-parameter set of orbits lying on $\mathbf{S}_{\mathrm{YM}}$. All these solutions end up at $\mathcal{P}_1$, except the heteroclinic orbit $\mathrm{FL}_0\rightarrow\mathrm{FL}_1$. 
	When $\gamma_{\mathrm{m}}<4/3$, only this heteroclinic orbit originates from $\mathrm{FL}_0$, while each point on $\mathrm{L}_\mathrm{R}$ has a center manifold ($\mathrm{L}_{\mathrm{R}}$ itself) and a two dimensional unstable manifold, being the source of a 1-parameter set of interior orbits (a 2-parameter set from the whole circle $\mathrm{L}_\mathrm{R}$). In this case, all solutions end at $\mathrm{FL}_1$ except the ones on $\mathbf{S}_{\mathrm{YM}}$ which end at $\mathcal{P}_1$. 
	If $\gamma_\mathrm{m}=4/3$, each fixed point on the disk $\mathrm{D}_\mathrm{R}$ is the source of a unique interior orbit. Since $\Omega_\mathrm{YM}=\text{const.}$, each periodic orbit $\mathcal{P}_{\Omega_{\mathrm{YM}}}$, at $T=1$, attracts a 1-parameter set of interior orbits, i.e. those solutions which originate from the circle of fixed points on the intersection of $\mathrm{D}_{\mathrm{R}}$ with $\Omega_\mathrm{YM}=\text{const.}$.
\end{proof}
\begin{figure}[h!]
\begin{center}
\subfigure[Solution space for $\gamma_{\text m}=1$.]{\label{fig:Dust}
\includegraphics[width=0.30\textwidth]{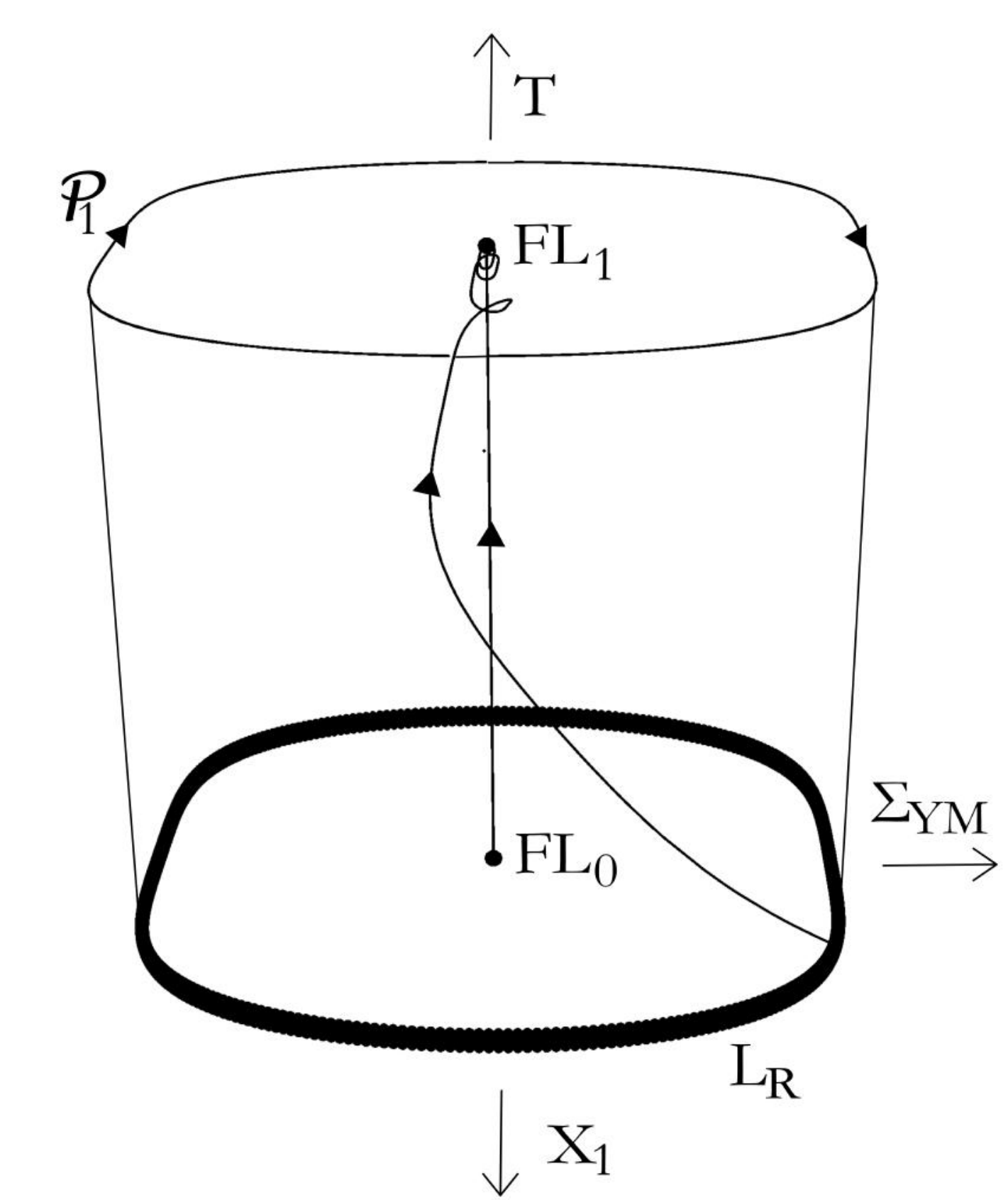}}
\subfigure[Solution space for $\gamma_{\text m}=\frac{4}{3}$.]{\label{fig:Radiation}
\includegraphics[width=0.30\textwidth]{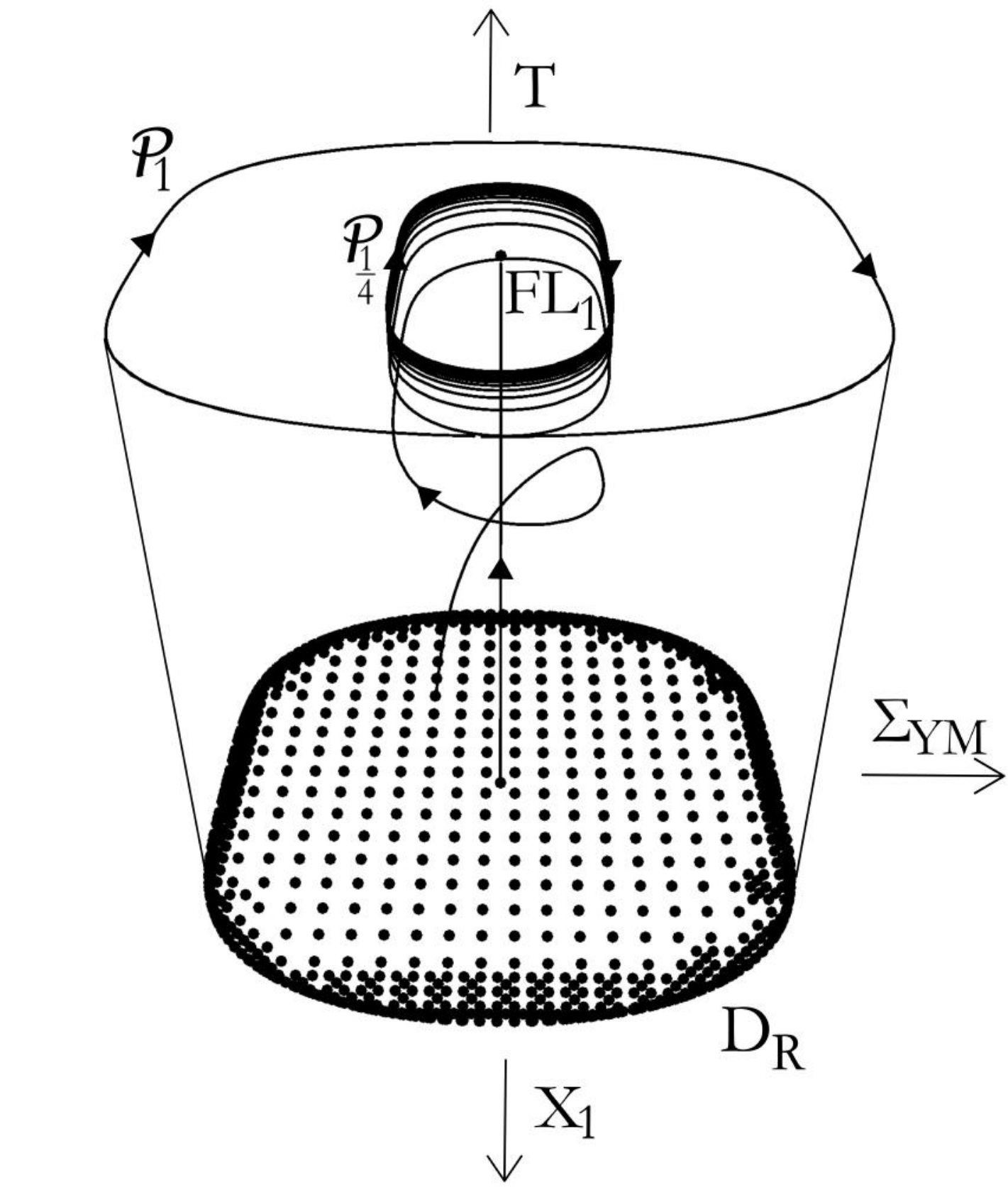}}
\subfigure[Solution space for $\gamma_{\text m}=\frac{3}{2}$.]{\label{fig:Stiffer}
\includegraphics[width=0.30\textwidth]{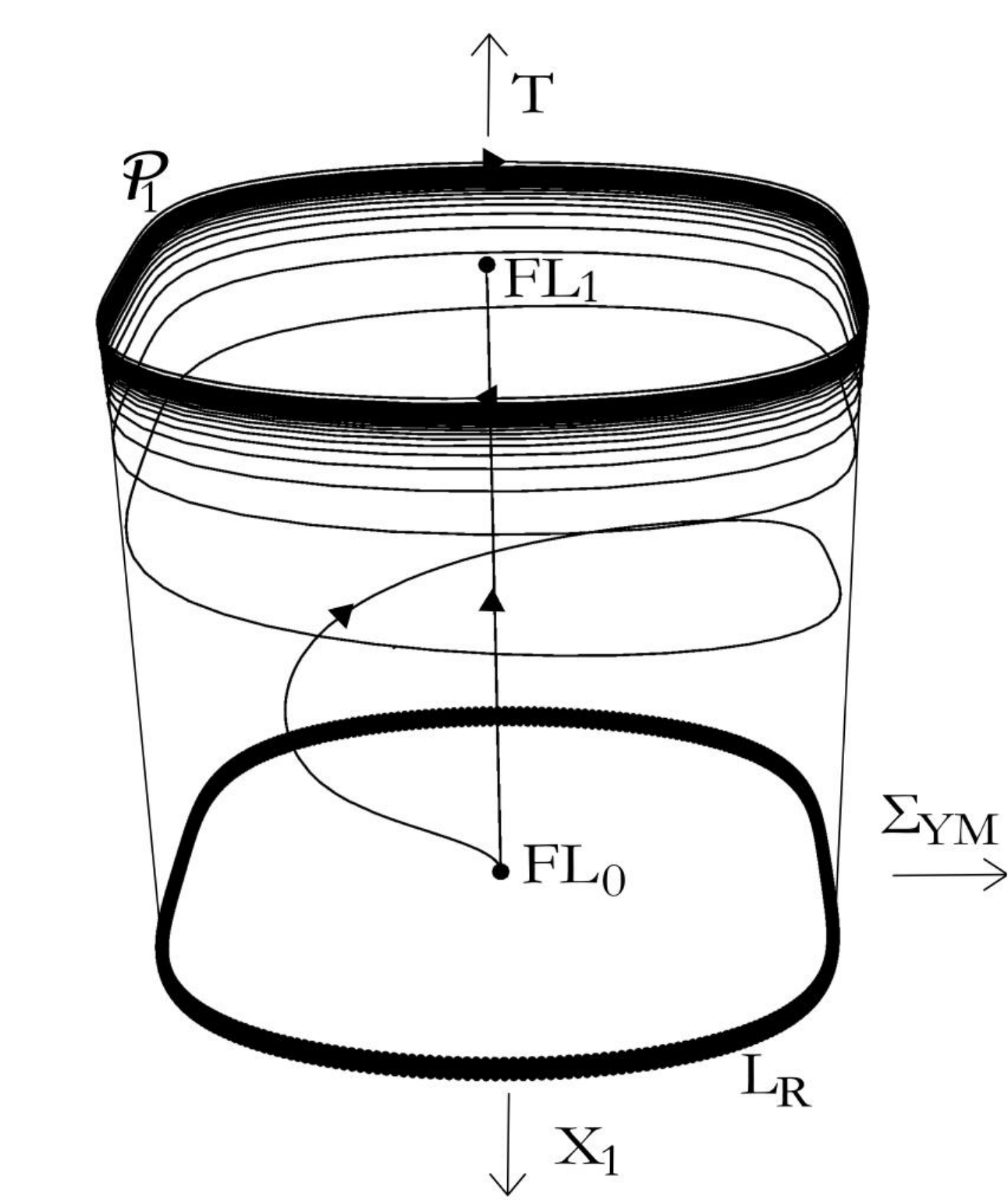}}
\end{center}
\vspace{-0.5cm}
\caption{Qualitative global evolution of dynamical system \eqref{DynSys2} in ${\overline{\bf S}}$ for the three different cases $\gamma_{\text m}<\frac{4}{3}$, $\gamma_{\text m}=\frac{4}{3}$ and $\gamma_{\text m}>\frac{4}{3}$, illustrating the results of Proposition \ref{prop1}.}
\label{fig:MasslessMatterSpace}
\end{figure} 
\section{Massive Yang-Mills field (case $\mu\ne 0$)}
\label{massive}
In this section, we analyze the system \eqref{eqspsi1} and \eqref{eqspsi-last}, with $\mu\ne 0$. We, therefore, introduce a new dimensionless variable associated with the mass parameter $\mu$,
\begin{equation}
X_2=\frac{\mu\phi}{\sqrt{2}H}.
\end{equation}
Using $e$-fold time $N$ as defined in~\eqref{e-folds}, we obtain the \emph{local} dynamical system
\begin{subequations} \label{prelim1}
 \begin{align}
  \frac{d\Sigma_{\mathrm{YM}}}{dN} 
   &=-\left[(1-q)\Sigma_{\mathrm{YM}}+2\tilde{T}X^3_1+\mu\tilde{T}^2 X_2\right] \\
   \frac{dX_1}{dN} &= -\frac{1}{2}\left[(1-q)X_1-2\tilde{T}\Sigma_{\mathrm{YM}}\right]\\
   \frac{dX_2}{dN} &=qX_2+\mu\tilde{T}^2\Sigma_{\mathrm{YM}} \\
  \frac{d\tilde{T}}{dN} &=\frac{1}{2}(1+q)\tilde{T},
 \end{align}
\end{subequations}
subject to the constraint
\begin{equation}\label{prelim-end}
X_2 = \mu\tilde{T} X_1, 
\end{equation}
and where we use
\begin{equation}\label{constrmass1}
1 - \Omega_{\text m} =\Sigma^{2}_{\mathrm{YM}}+X^4_1+X^2_2
\end{equation}
to solve for $\Omega_\mathrm{m}$. The deceleration parameter $q$ is given by
\begin{equation}
\label{q-massive-qwe}
 q=1-X^2_2+\frac{3}{2}\left(\gamma_{\text m}-\frac{4}{3}\right)\Omega_{\text m}.
\end{equation}
As in the massless case, the constraint~\eqref{constrmass1} implies that $X_1$, $\Sigma_{\mathrm{YM}}$, $\Omega_\mathrm{m}$, and $X_2$ are bounded. In particular, the bounds in~\eqref{BoundVar} hold and, in addition,
\begin{equation}
-1\leq X_2 \leq 1,
\end{equation}
which, given $0<\gamma_{\mathrm{m}}\leq 2$ and \eqref{q-massive-qwe}, yields
\begin{equation}
-1<q\leq 2 \, .
\end{equation}
Since the constraint~\eqref{prelim-end} is linear in $X_2$, it  can be used to solve for $X_2$ giving a \emph{local} 3-dimensional dynamical system for $(X_1,\Sigma_{\mathrm{YM}},\tilde{T})$, which is particularly useful for analyzing the asymptotics when  $H\rightarrow+\infty$ ($\tilde{T}\rightarrow0$), where $X_2\rightarrow0$. 
One could, as well, construct a local dynamical systems formulation appropriated to study the dynamics when $\tilde{T}$ becomes unbounded, i.e. $H\rightarrow 0$. This can be achieved by replacing $\tilde{T}$ with $\bar{T}=\tilde{T}^{-1}$, together with a new time variable $\tilde{N}$ defined via $d/d\tilde{N}=\bar{T}^3d/dN$ and where, now, the constraint becomes linear in $X_1 =\mu^{-1}\bar{T}X_2$ and, hence, can be solved for $X_1$ to obtain a {\em local} dynamical system for $(X_2,\Sigma_{\mathrm{YM}},\bar{T})$, with $X_1\rightarrow0$ as $\bar{T}\rightarrow0$. 

To obtain a {\em global} dynamical systems formulation on a compact state space, we proceed as in the massless case, and introduce the bounded variable
\begin{equation}
 T=\frac{\tilde{T}}{1+\tilde{T}}~,
\end{equation}
which satisfies $0<T<1$. By introducing a new independent variable $\bar \tau$, such that 
\begin{equation}
 \frac{d}{d\bar{\tau}}=(1-T)^2\frac{d}{dN},
\end{equation}
we obtain, from \eqref{prelim1} and \eqref{prelim-end}, a \emph{global} dynamical system
\begin{subequations} 
\label{dyn1-massive}
 \begin{align}
 \label{dyn-massive1}
 \frac{d\Sigma_{\mathrm{YM}}}{d\bar{\tau}} &=-\left[(1-q)(1-T)^2 \Sigma_{\mathrm{YM}}+2(1-T)T X^3_1+\mu T^2 X_2\right] \\
  \frac{dX_1}{d\bar{\tau}} &= -\frac{1}{2}\left[(1-q)(1-T)^2 X_1-2(1-T)T\Sigma_{\mathrm{YM}}\right]\\
  \frac{dX_2}{d\bar{\tau}} &=q(1-T)^2 X_2+\mu T^2\Sigma_{\mathrm{YM}} \\
  \label{T-evolution-massive}
  \frac{dT}{d\bar{\tau}} &=\frac{1}{2}(1+q)(1-T)^3 T,
 \end{align}
\end{subequations}
subject to the constraint
\begin{equation}\label{massconstr}
(1-T) X_2 = \mu T X_1,
\end{equation}
and where we use~\eqref{constrmass1} to globally solve for $\Omega_{\mathrm{m}}$. The deceleration parameter $q$ is, then, given by
\begin{equation}
\label{q-massive}
q=1-X^2_2+\frac{3}{2}\left(\gamma_{\text m}-\frac{4}{3}\right) (1 - \Sigma^{2}_{\mathrm{YM}}-X^4_1-X^2_2)  \\.
\end{equation}
It is also useful to consider the \emph{auxiliary} evolution equation for the effective energy density of the Yang-Mills field which, in the present case, reads
\begin{equation}
\Omega_{\mathrm{YM}} :=\frac{\rho_{\mathrm{YM}}}{3H^2} = \Sigma^{2}_{\mathrm{YM}}+X^4_1+X^2_2,
\end{equation}
with $\Omega_{\mathrm{YM}} =1-\Omega_{\text m}$. From \eqref{rhoYM} and \eqref{gammafunction}, we can write
\begin{equation}
\gamma_{\mathrm{YM}}:= 1+\frac{p_{\mathrm{YM}}}{\rho_{\mathrm{YM}}} =1+\frac{1}{3}\frac{\Sigma^{2}_{\mathrm{YM}}+X^{4}_{1}-X^{2}_{2}}{\Omega_{\mathrm{YM}}}
=\frac{4}{3}-\frac{2}{3}\frac{X^{2}_{2}}{\Omega_{\mathrm{YM}}}.
\end{equation}
Furthermore, rewriting \eqref{q-massive} as
\begin{equation}
q = -1+\frac{3}{2}\left(\gamma_{\mathrm{YM}}\Omega_{\mathrm{YM}}+\gamma_{\text m}\Omega_{\text m}\right),
\end{equation}
we obtain
\begin{equation}
\frac{d\Omega_{\mathrm{YM}}}{d\bar{\tau}}=3(1-T)^2(\gamma_{\text m}\Omega_{\mathrm{YM}}-\gamma_{\mathrm{YM}}\Omega_{\mathrm{YM}})(1-\Omega_{\mathrm{YM}}).
\end{equation}
The price to pay, in order to have a global relatively compact state space picture, is that the constraint~\eqref{massconstr} cannot be solved globally. However, it forms an invariant set for the flow. This can be seen by writing $G(X_1,X_2,T)=(1-T)X_2-\mu TX_1=0$ and noticing that
\begin{equation}
\frac{dG}{d\bar{\tau}}=\left(q-\frac{1}{2}(1+q)T\right)(1-T)^2 G\, .
\end{equation}
The state-space ${\bf S}$ for the variables $(\Sigma_{\mathrm{YM}}, X_1, X_2,  T)$ is, therefore, the subset defined by $G=0$ on the set $\{0\leq \Sigma^2_{\mathrm{YM}}+X^4_1+X^2_2\leq1 \wedge 0<T<1 \}$.
The state-space ${\bf S}$ contains other important invariant subsets: the pure Yang-Mills subset $\Omega_{\mathrm{m}}=0$ and the Friedmann-Lema\^itre invariant subset for which $\Omega_{\mathrm{m}}=1$. In addition, it can be regularly extended to include the invariant boundaries $T=0$ and $T=1$ to obtain the compact state-space $\overline {\bf S}$. 

As a starting point for our analysis, we study the past and future limit sets: 
\begin{lemma} \label{lemma-massive}Consider the system \eqref{dyn1-massive}-\eqref{massconstr}.
	The $\alpha$-limit set of all interior orbits in $\mathbf{S}$ is located at the invariant boundary $T=0$, with $X_2=0$, and 
	the $\omega$-limit set is at the invariant boundary $T=1$, with $X_1=0$.
\end{lemma}

\begin{proof}
	We make use of the \emph{monotonicity principle}. Due to $1+q>0$, a quick inspection of equation~\eqref{T-evolution-massive}
	reveals that $T(\bar{\tau})$ is monotonically increasing in {\bf S} and, therefore, there are no periodic nor recurrent orbits in the interior of the state space. We conclude that the $\alpha$-limit set of all solutions is located at the $T=0$ invariant boundary associated to the asymptotic past $H\rightarrow +\infty$, while the $\omega$-limit set is located on the $T=1$ invariant boundary, associated with the asymptotic future $H\rightarrow0$. Moreover, the constraint~\eqref{massconstr} implies that $X_1\rightarrow0$ as $T\rightarrow 1$, and $X_2\rightarrow 0$ as $T \rightarrow 0$.
\end{proof}

\begin{remark}
	This lemma implies, in particular, the result in~\cite{bentoetal93}, that for the pure Yang-Mills field, the past asymptotics is dominated by the "massless potential", while the future asymptotics it is dominated by the "mass potential".
\end{remark}
We now proceed with a detailed analysis of the past and future asymptotics. 
\subsection*{A. Past asymptotics for massive Yang-Mills fields and perfect fluids}
%
Since along $G=0$, we have $X_2\rightarrow0$ as $T\rightarrow0$, then the invariant boundary $T=0$ coincides with the $T=0$ boundary of the massless Yang-Mills state space. It follows that there exist the fixed points $\mathrm{FL}_0$, as well as the deformed circle L$_\mathrm R$ and the disk D$_\mathrm R$ of fixed points now for:
\begin{align}
\mathrm{D_R}:&\qquad 0\leq\Sigma^2_{\mathrm{YM}}+X^4_1\leq 1,~T=0,~X_2=0\\ 
\text{L}_{\text{R}}:&\qquad \Sigma^2_{\mathrm{YM}}+X^4_1=1,~T=0,~X_2=0\\ 
{\text{FL}}_0:&\qquad \Sigma_{\mathrm{YM}} = X_1 = 0,~T=0, ~X_2=0\,.
\end{align}
The goal of this subsection is to prove the next theorem which gives a description of the past asymptotics of the model.
\begin{theorem} Consider solutions of the system \eqref{dyn1-massive} and \eqref{massconstr} with $0<\Omega_{\mathrm m}<1$. For $\bar{\tau}\rightarrow-\infty$:
 \begin{itemize}
 	\item[(i)] If $0<\gamma_{\mathrm m} < \frac{4}{3}$, all solutions converge to the circle of fixed points $\mathrm{L_R}$. More precisely, each fixed point on $\mathrm{L}_{\mathrm{R}}$ is the $\alpha$-limit point of a 2-parameter set of solutions. 
 	\item[(ii)] If $\gamma_{\mathrm m} = \frac{4}{3}$, all solutions converge to the disk of fixed points $\mathrm{D_R}$. More precisely, each fixed point on $\mathrm{D}_{\mathrm{R}}$ is the $\alpha$-limit point of a unique solution.
  \item[(iii)] If $\gamma_{\mathrm m} > \frac{4}{3}$, all solutions  converge to the fixed point $\mathrm{FL_0}$.
   \end{itemize}
\end{theorem}
\begin{proof}
The proof uses Lemma~\ref{lemma-massive} and the fact that $X_2=0$ at $T=0$, which means that the orbit structure on this boundary coincides with that of the massless case studied  in Subsection~\ref{T1}~B. In fact, this boundary consists of heteroclinic orbits when $\gamma_{\mathrm{m}}\neq \frac{4}{3}$, or only of fixed points when $\gamma_{\mathrm{m}}= \frac{4}{3}$, see Figure~\ref{fig:T0_Massless}.  Therefore, the possible past attracting sets are fixed points located at $T=0$. In order to deduce the stability properties of the fixed points, we need to solve the constraint~\eqref{massconstr}. Although it is not possible to solve this constraint globally, we can uniquely solve it locally at the points where $\nabla G\neq0$ by making use of the implicit function theorem. Since $\partial_{X_2}G|_{T=0}=1$, in a neighborhood of the $T=0$ boundary, then we can eliminate the variable $X_2$ from the eigenvalue analysis of the fixed points on $T=0$.  This yields the same results as the linearisation around the corresponding similar fixed points of the massless case.
\end{proof}
The physical interpretation of the above theorem is that, if $\gamma_\mathrm{m}<4/3$, the dynamics are past asymptotically dominated by the massless Yang-Mills field while, if the fluid content has an equation of state stiffer than radiation, the past asymptotics is governed by the Friedmann-Lema\^itre solution. If $\gamma_{\mathrm{m}}=4/3$, then the model is neither fluid of massless Yang-Mills dominated towards the past. 
%
\subsection*{B. Future asymptotics for massive Yang-Mills fields and perfect fluids}
We start by describing the future invariant subset $T=1$. Since $X_1=0$ at $T=1$, the induced flow on this boundary is given by
\begin{subequations}
	\begin{align}
	\frac{d\Sigma_{\mathrm{YM}}}{d\bar{\tau}} &=- \mu X_2 \label{1} \\
	\frac{dX_2}{d\bar{\tau}} &=\mu \Sigma_{\mathrm{YM}}  \label{2}~,
	\end{align}
\end{subequations}
where now
\begin{equation}
	1 - \Omega_{\mathrm m}= \Sigma^{2}_{\mathrm{YM}}+X^2_2.
\end{equation}
The $T=1$ boundary is foliated by periodic orbits $\mathcal{P}_{\Omega_{\mathrm{YM}}}$, parametrized by constant values of $\Omega_{\mathrm{YM}}=\Sigma^2_{\mathrm{YM}}+X^2_2$, with the fixed point $\mathrm{FL}_1$ given by
\begin{equation}
{\text{FL}}_1:\qquad \Sigma_{\mathrm{YM}} = X_2 = 0,~T=1, ~X_1=0
\end{equation}
and located at the center, see Fig.~\ref{T1massive}.
\begin{figure}[ht!]
	\begin{center}
			\includegraphics[width=0.30\textwidth]{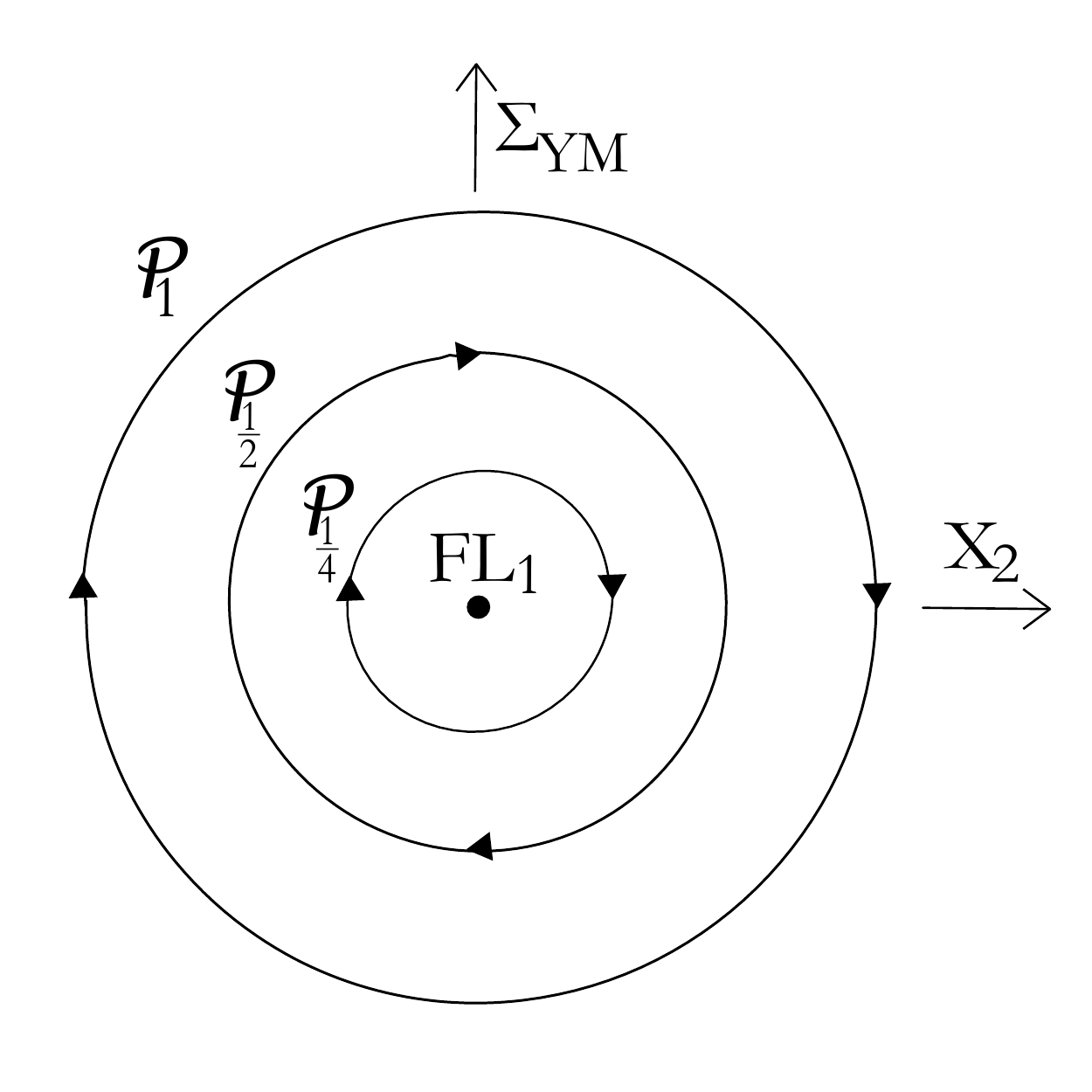}
		\end{center}
	\vspace{-0.5cm}
	\caption{Representation of the invariant boundary $T=1$ when $\mu>0$.}
	\label{T1massive}
\end{figure}
The objective of this subsection is to prove the following result:
\begin{theorem}  Consider solutions of the system \eqref{dyn1-massive} and \eqref{massconstr} with $0<\Omega_{\mathrm{m}}<1$. For $\bar{\tau}\longrightarrow+\infty$:
	\label{theorem}
 \begin{itemize}
  \item[(i)] If $\gamma_{\mathrm m}>1$, then all solutions converge  to the outer periodic orbit $\mathcal{P}_{1}$ with $\Omega_\mathrm{m}=0$. 
  \item[(ii)] If $\gamma_{\mathrm m}<1$, then all solutions converge to fixed point $\mathrm{FL_1}$ with $\Omega_\mathrm{m}=1$.
  \item[(iii)] If $\gamma_{\mathrm m}=1$, then a 1-parameter set of solutions converge to each inner periodic orbit $\mathcal{P}_{\Omega_{\mathrm{YM}}}$.
   \end{itemize}
\end{theorem}

\begin{proof}
The proof is based on Lemma~\ref{lemma-massive} together with averaging techniques and consists of an adaptation of the methods used in Ref.\cite{alhetal15}.
An important difference with respect to the standard averaging theory is that the perturbation parameter $\varepsilon$ will not be a constant, but a function of time here. We start by recalling that each periodic orbit on $T=1$ has an associated time period $P(\Omega_{\mathrm{YM}})$, so that, for a given function $f$, its average over a time period 
characterized by $\Omega_{\mathrm{YM}}$ is given by 
\begin{equation}
 \langle f \rangle_{\Omega_{\mathrm{YM}}} = \frac{1}{P(\Omega_{\mathrm{YM}})} \int^{\bar{\tau}_0+P(\Omega_{\mathrm{YM}})}_{\bar{\tau}_0} f(\bar{\tau})\,d\bar{\tau}.
\end{equation}
Differentiating \eqref{2} and using \eqref{1} gives
\begin{equation}
 \frac{d}{d\bar{\tau}}\left(X_2 \frac{dX_2}{d\bar{\tau}} \right)-\left(\frac{dX_2}{d\bar{\tau}}\right)^2 + \mu^2 X^2_2 =0.
\end{equation}
Taking the average for a periodic orbit gives,
\begin{equation}
 \left\langle \left(\frac{dX_2}{d\bar{\tau}}\right)^2 \right\rangle =  \mu^2 \langle X^2_2 \rangle, 
\end{equation}
which implies
\begin{equation}
 \langle \Sigma^2_{\mathrm{YM}} \rangle = \langle X^{2}_{2} \rangle.
\end{equation}
Thus, on the $T=1$ invariant subset
\begin{equation}\label{gammaAv}
 \langle \gamma_{\mathrm{YM}} \rangle = \frac{4}{3}-\frac{2}{3}\frac{\langle X^2_2 \rangle}{\langle \Sigma^2_{\mathrm{YM}}\rangle + \langle X^{2}_2 \rangle}=1,
\end{equation}
which does not depend on $\Omega_{\mathrm{YM}}$ and, on average, the Yang-Mills field  behavior resembles that of dust. 

We now set $\varepsilon(\bar{\tau})=1-T(\bar{\tau})$ and consider the system 
\begin{subequations}
 \begin{align}
 \label{evol-omega}
 \frac{d\Omega_{\mathrm{YM}}}{d\bar{\tau}} &= 3\varepsilon^2 \left(\gamma_{\mathrm{m}}\Omega_{\mathrm{YM}} -\gamma_{\mathrm{YM}}\Omega_{\mathrm{YM}}\right)\left(1-\Omega_{\mathrm{YM}}\right) := \varepsilon^2 f(\Omega_{\mathrm{YM}},\bar{\tau},\varepsilon) \\
 \label{evol-epsi}
  \frac{d\varepsilon}{d\bar{\tau}} &= -\frac{1}{2}(1+q)\varepsilon^3 (1-\varepsilon),
 \end{align}
\end{subequations}
where
\begin{equation}
\gamma_{\mathrm{YM}}\Omega_{\mathrm{YM}}=\frac{4}{3}\Omega_{\mathrm{YM}}-\frac{2}{3}X^{2}_{2}\quad,\quad 1 + q = \frac{3}{2}(\gamma_{\mathrm{m}}-(\gamma_{\mathrm{m}}-\gamma_{\mathrm{YM}})\Omega_{\mathrm{YM}})
\end{equation}
and $(X_1,X_2,\Sigma_{\mathrm{YM}})$ solves~\eqref{dyn1-massive} and \eqref{massconstr} with the equation for $T$ replaced by the equation for $\varepsilon$. Recall that $1+q>0$ and, therefore, $\varepsilon$ is monotonically decreasing, so that, $\varepsilon(\bar{\tau})\rightarrow0$ as $\bar{\tau}\rightarrow+\infty$. Moreover, since $\partial_{X_1} G|_{T=1}=\mu\neq0$, we can use the implicit function theorem to solve~\eqref{massconstr} uniquely for $X_1$, in a neighborhood of the $T=1$ boundary. 

We start by applying the near-identity transformation depending on $\varepsilon$,
\begin{equation}\label{NearId}
\Omega_{\mathrm{YM}}(\bar{\tau}) = y(\bar{\tau}) + \varepsilon^2(\bar{\tau}) w(y,\bar{\tau},\varepsilon).
\end{equation}
%
The evolution equation for $y$ is obtained using \eqref{evol-omega} and \eqref{evol-epsi}, which gives
\begin{equation}
\begin{split}
\label{y-evolution}
\frac{dy}{d\bar{\tau}} &= \left(1+\varepsilon^2\frac{\partial w}{\partial y}\right)^{-1}\left[\frac{d\Omega_{\mathrm{YM}}}{d\bar{\tau}}-\left(2\varepsilon w+\varepsilon^2\frac{\partial w}{\partial\varepsilon} \right)\frac{d\varepsilon}{d\bar{\tau}} -\varepsilon^2\frac{\partial w}{\partial\bar{\tau}} \right] \\
&=\frac{\varepsilon^2}{1+\varepsilon^2\frac{\partial w}{\partial y}} \Big[3(\gamma_{\mathrm m}-1)y(1-y)+3(1-\gamma_{\mathrm{YM}})y(1-y)+3w\varepsilon^4(\gamma_{\mathrm m}-\gamma_{\mathrm{YM}}) + 3\varepsilon^6 (\gamma_{\mathrm m}-\gamma_{\mathrm{YM}})- \\
 &- \frac{\partial w}{\partial\bar{\tau}}+\left(2w+\varepsilon\frac{\partial w}{\partial\varepsilon} \right)\left(\frac{1+q}{2}\right)(1-\varepsilon)\varepsilon^2\Big]
\end{split}
\end{equation}
and where we used~\eqref{gammaAv}. Setting
\begin{equation} \label{choice}
\begin{split}
\frac{\partial w}{\partial\bar{\tau}} 
&= f(y,\bar{\tau},\varepsilon)-\langle f(y,\cdot,0)\rangle \\
&=3(1-\gamma_{\mathrm{YM}}) y (1-y) \\
&= \left( - y + 2 X^{2}_{2} \right)(1-y)
\end{split}
\end{equation}
and expanding \eqref{y-evolution} in powers of $\varepsilon$, for $\varepsilon$ sufficiently small, the equation for $\Omega_{\mathrm{YM}}$ is transformed into the full averaged equation 
\begin{equation}\label{fullAv}
\frac{dy}{d\bar{\tau}} ={\varepsilon^2} \langle f \rangle(y) + {\varepsilon^4 h(y,w,\bar{\tau},\varepsilon)}+\varepsilon^{5}(1+q)\left(\frac{1}{2}\frac{\partial w}{\partial \varepsilon}-w\right)+ \mathcal{O}({\varepsilon^{6}}),
\end{equation}
where
\begin{eqnarray}
\langle f\rangle(y)&=&\langle f(y,\cdot,0) \rangle = 3(\gamma_{\mathrm m}-1)y(1-y) \\
h(y,w,\bar{\tau},\varepsilon) &= & w(1+q)+3w(1-2y)(\gamma_{\mathrm m}-\gamma_{\mathrm{YM}})(1-2y)-3(\gamma_{\mathrm m}-\gamma_{\mathrm{YM}})\frac{\partial w}{\partial y}y(1-y). \label{hdef}
\end{eqnarray}
Note that, due to the previous analysis of the invariant set $T=1$, i.e. $\varepsilon=0$, the right-hand-side of~\eqref{choice} is, for large times, almost-periodic and has zero mean, which, in particular, implies that $w$ is bounded.
Then, it follows from~\eqref{NearId} that $y$ is also bounded. Moreover, for sufficiently small $\varepsilon$, Eq.~\eqref{fullAv} implies that $y$ is monotonic, either increasing or decreasing depending on the sign of $\gamma_{\mathrm{m}}-1\neq0$ and, hence, $y$ has a limit when $\bar{\tau}\rightarrow+\infty$.

Now, we study the evolution of the truncated averaged equation, which is obtained by dropping all higher order terms in~\eqref{fullAv} as
\begin{eqnarray}
\frac{d\bar{y} }{d\bar{\tau}}  &=& 3\varepsilon^2(\gamma_{\mathrm{m}}-1)\bar{y}(1-\bar{y}) \label{orisys1} \\
\frac{d\varepsilon}{d\bar{\tau}}  &=&-\frac{1}{2}(1+q)(1-\varepsilon)\varepsilon^3. \label{orisys2}
\end{eqnarray}
In this system, the $\varepsilon=0$ axis consists of a non-hyperbolic line of fixed points. 
Making the change of time variable 
$$\frac{1}{\varepsilon^2}\frac{d}{d\bar{\tau}}=\frac{d}{d\tilde{\tau}},$$ 
which does not affect the behavior of interior orbits, i.e. orbits with $\varepsilon>0$, we get
\begin{eqnarray}
\frac{d\bar{y} }{d\tilde{\tau}} &=& 3(\gamma_{\mathrm{m}}-1)\bar{y} (1-\bar{y} ) \\
\frac{d\varepsilon}{d\tilde{\tau}} &=& -\frac{1}{2}\varepsilon(1+q)(1-\varepsilon).
\end{eqnarray}
For $\gamma_{\mathrm{m}}-1\neq 0$, the above dynamical system has the two fixed points $\mathrm{P}_1=(\bar{y}=0;\varepsilon=0)$ and $\mathrm{P}_2=(\bar{y}=1;\varepsilon=0)$, where the $\varepsilon=0$ axis consists now of the heteroclinic orbit $\mathrm{P}_1\rightarrow \mathrm{P}_2$ (resp. $\mathrm{P}_2\rightarrow \mathrm{P}_1$), in case $\gamma_{\mathrm{m}}-1 > 0$ (resp. $\gamma_{\mathrm{m}}-1 < 0$).
Thus, for $\gamma_{\mathrm{m}}>1$ (resp. $\gamma_{\mathrm{m}}<1$), solution trajectories of the system~\eqref{orisys1} and \eqref{orisys2} will converge to the fixed point $\mathrm{P}_2$ (resp. $\mathrm{P}_1$), tangentially to the $\varepsilon=0$ axis.

Next, we show that solutions $y$, of the full averaged Eq. \eqref{fullAv}, have the same limit as the solutions $\bar{y}$ of the truncated averaged equation when $\bar{\tau}\rightarrow +\infty$. For this, we define the sequences $\{\bar\tau_n\}$ and $\{\varepsilon_n\}$ such that $\varepsilon_n=\varepsilon(\bar\tau_n)$, with $n\in \mathbb{N}$, and
\begin{eqnarray}
\bar{\tau}_{n+1}-\bar{\tau}_n &=& \frac{1}{\varepsilon_n^2}\\
\bar{\tau}_0 &=& 0 \\
\varepsilon_0&>&0, 
\end{eqnarray}
where $\lim \bar{\tau}_n = +\infty$ and 
$\lim \varepsilon_n=0$, since $\varepsilon(\bar{\tau})\rightarrow0$ as $\bar{\tau}\rightarrow+\infty$. We estimate $|\eta(\bar{\tau})|=|y(\bar{\tau})-\bar{y}(\bar{\tau})|$ as follows
\begin{equation}
\begin{split}
|\eta(\bar{\tau})|&=\left|\int_{\bar{\tau}_n}^{\bar{\tau}}\left(3\varepsilon^2(\gamma_{\mathrm{m}}-1)y(1-y)+\varepsilon^4h(y,w,\varepsilon,s)\right)ds-\int_{\bar{\tau}_n}^{\bar{\tau}}3\varepsilon^2(\gamma_{\mathrm{m}}-1)\bar{y}(1-\bar{y})ds + \mathcal{O}(\varepsilon^5)\right| \nonumber \\
&\leq \varepsilon^2\int_{\bar{\tau}_n}^{\bar{\tau}}3\underbrace{|\gamma_{\mathrm{m}}-1|}_{|\cdot|\leq C} |(y-\bar{y})\underbrace{(1-(y+\bar{y}))}_{|\cdot|\leq 1}|ds+\varepsilon^4\int_{\bar{\tau}_n}^{\bar{\tau}} \underbrace{|h(y,w,\varepsilon,s)|}_{|\cdot|\leq M} ds+ \mathcal{O}(\varepsilon^5) \\
&\leq 3C\varepsilon_n^2\int_{\bar{\tau}_n}^{\bar{\tau}}|\eta(s)|ds+\varepsilon_n^4M(\bar{\tau}-\bar{\tau}_n)+\mathcal{O}(\varepsilon^5_n), \nonumber
\end{split}
\end{equation} 
where $C$ and $M$ are some positive constants. By Gronwall's inequality
\begin{equation}
|\eta(\bar{\tau})|\leq \frac{\varepsilon_n^2 M}{3C}(e^{3C\varepsilon_n^2(\bar{\tau}-\bar{\tau}_n)}-1)+\mathcal{O}(\varepsilon^3_n),
\end{equation}
and using the fact that $\bar{\tau}-\bar{\tau}_n\in [0,1/\varepsilon_n^2]$, we find
\begin{equation}
|\eta(\bar{\tau})|\leq K\varepsilon_n^2,
\end{equation}
with $K$ a positive constant. As $\varepsilon_n\rightarrow 0$,  then $|\eta(\bar{\tau})|\rightarrow0$, and so $y$ and $\bar{y}$ have the same limit. Finally, from equation~\eqref{NearId}, the triangular inequality, and the fact that $\varepsilon\rightarrow0$ as $\bar{\tau}\rightarrow+\infty$, it follows that $\Omega_{\mathrm{YM}}$ has the same limit as $\bar{y}$ and, therefore, converges to $0$ or $1$, depending on the sign of $\gamma_{\mathrm{m}}-1\neq0$. This proves cases $(i)$ and $(ii)$ of the theorem.
\\\\
Now, we analyses the case when $\gamma_{\mathrm{m}}=1$. In that case, the equation for $y$ is given by
\begin{equation}
\frac{dy}{d\bar{\tau}}=\varepsilon^4 h(y,w,\varepsilon,\bar{\tau})+\mathcal{O}(\varepsilon^5).
\end{equation}
Taking the average of $h$, given in~\eqref{hdef}, at $\varepsilon=0$, 
\begin{eqnarray}
\langle h\rangle(y,w) &=&\langle h(y,\cdot,0)\rangle = \frac{1}{P}\int_{0}^{P}h(y,w,0,\bar{\tau})d\bar{\tau} \nonumber \\
&=&  \frac{1}{P}\int_{0}^{P}w(y,\bar{\tau},0)(1+q)d\bar{\tau} \nonumber \\
&=& \frac{1}{P}\int_{0}^{P}\frac{3}{2}w(y,\bar{\tau},0)(1+(\gamma_{\mathrm{YM}}-1)y)d\bar{\tau} \nonumber \\
&=& \frac{3}{2}\langle w(y,\cdot,0)\rangle =\frac{3}{2}\langle w\rangle(y),
 \end{eqnarray}
we consider the truncated averaged equation
\begin{eqnarray}
\frac{d\bar{z}}{d\bar{\tau}}&=&\frac{3}{2} \varepsilon^4 \langle w\rangle(\bar{z}) \label{AverCrit1}\\
\frac{d\varepsilon}{d\bar{\tau}}&=&-\frac{3}{4}\varepsilon^4(1-\varepsilon)\,. \label{AverCrit2}
\end{eqnarray}
To resolve the non-hyperbolicity of the line of fixed points at $\varepsilon=0$, we  make the change of time variable $\varepsilon^{-3}d/d\bar{\tau}=d/d\tilde{\tau}$, to obtain
\begin{eqnarray}
\frac{d\bar{z}}{d\tilde{\tau}}&=&\frac{3}{2} \varepsilon \langle w\rangle(\bar{z}) \\
\frac{d\varepsilon}{d\tilde{\tau}}&=&-\frac{3}{4}\varepsilon(1-\varepsilon).
\end{eqnarray}
In this case, the $\varepsilon=0$ axis consists of a line of fixed points with $\bar{z}_0\in[0,1]$, whose linearisation yields the eigenvalues $\lambda_1=0$ and $\lambda_2=-\frac{3}{4}$ with associated eigenvectors $v_1=(\bar{z}=1;\varepsilon=0)$ and $v_2=(\bar{z}=-2\langle w\rangle(\bar{z}_0);\varepsilon=1)$. Therefore, the line is normally hyperbolic and each point on the line is exactly the $\omega$-limit point of a unique interior orbit. This means that there also exists an orbit of the dynamical system~\eqref{AverCrit1} and \eqref{AverCrit2} with $\varepsilon>0$ initially, that converges to $(\bar{z}_0,0)$, for each $\bar{z}_0$, as $\tilde{\tau} \rightarrow +\infty$. 

Just as in the proof of cases $(i)$ and $(ii)$,  we can estimate the term $\mathcal{O}(\varepsilon^5)$ that provides bootstraping sequences. This defines a pseudo-trajectory $\Omega^n_{\mathrm{YM}}(\bar{\tau}_n)=\bar{z}(\bar{\tau}_n)$ of system~\eqref{evol-omega} and \eqref{evol-epsi}, with
\begin{equation}
|\Omega^n_{\mathrm{YM}}(\bar{\tau})-\bar{z}(\bar{\tau})|\leq K\varepsilon_n^2 \,,
\end{equation}
where $\bar{\tau} \in [\bar{\tau}_n,\bar{\tau}_{n+1}]$ and $K$ is a positive constant. Compactness of the state space and the regularity of the flow implies that exists a set of initial values whose solution trajectory $\Omega_{\mathrm{YM}}(\bar{\tau})$ shadows the pseudo-trajectory $\Omega^n_{\mathrm{YM}}(\bar{\tau})$, in the sense that
\begin{equation}
\forall n\in \mathbb{N},~~\forall\bar{\tau}\in[\bar{\tau}_n,\bar{\tau}_{n+1}]:~~|\Omega^n_{\mathrm{YM}}(\bar{\tau})-\Omega_{\mathrm{YM}}(\bar{\tau})|\leq K\varepsilon_n^2 \,.
\end{equation}
Finally, using the triangle inequality, we get
\begin{eqnarray}
|\Omega_{\mathrm{YM}}(\bar{\tau})-\bar{z}(\bar{\tau})|&=&|\Omega_{\mathrm{YM}}(\bar{\tau})-\Omega^n_{\mathrm{YM}}(\bar{\tau})+\Omega^n_{\mathrm{YM}}(\bar{\tau})-\bar{z}(\bar{\tau})| \nonumber \\
&\leq&\underbrace{|\Omega^n_{\mathrm{YM}}(\bar{\tau})-\Omega_{\mathrm{YM}}(\bar{\tau})|}_{\leq K\varepsilon_n^2}+\underbrace{|\Omega^n_{\mathrm{YM}}(\bar{\tau})-\bar{z}(\bar{\tau})|}_{\leq K\varepsilon_n^2} \nonumber \\
&\leq& 2K\varepsilon_n^2 \underbrace{\rightarrow}_{\bar{\tau}_n \rightarrow \infty} 0 \, ,
\end{eqnarray}
and, therefore, for each $\bar{z}_0 \in [0,1]$, there exists a solution trajectory $\Omega_{\mathrm{YM}}(\bar{\tau})$ that converges to a  periodic orbit at $\varepsilon=0$ i.e. $T=1$, characterized by $\Omega_{\mathrm{YM}}=\bar{z}_0$, which concludes the proof of $(iii)$.
\end{proof}
The physical interpretation of the above theorem is that, if $\gamma_\mathrm{m}<1$, then the general solutions of the massive system behave like the Friedmann-Lema\^itre solution asymptotically towards the future. However, if the fluid content has an equation of state stiffer than dust, then the future asymptotics is governed by the pure massive Yang-Mills  solution, which, in particular, exhibits oscillatory behavior If $\gamma_{\mathrm{m}}=1$, then the model is neither fluid of massive  Yang-Mills dominated towards the future. 
%
\section{Concluding Remarks} \label{CR}
The present paper considers spatially homogeneous and isotropic massless and massive Yang-Mills field cosmologies with a perfect fluid. In particular the well-known explicitly solvable massless Yang-Mills isotropic cosmologies~\cite{bentoetal93,GV91} have been contextualized in a global dynamical systems formulation on a compact state-space.

The above dynamical systems formulations can be used to shed light on the dynamics of more general anisotropic cosmological models, where massless Yang-Mills fields are known to exhibit past asymptotic chaotic behavior reminiscent of the Mixmaster universe as well as future asymptotic oscillatory behavior similar to Yang-Mills field in Minkowski space~\cite{BYM05,DK96,BL98,YM05,DK97}. General spatially homogeneous Yang-Mills fields under the Hamiltonian gauge can be written in diagonal form 
$A^{a}_{i}=\chi_{(i)}(t) \delta^{a}_{i}$, 
where for the diagonal Bianchi class A, if the off-diagonal components are zero initially, then they will remain so for the whole evolution. Isotropy requires all diagonal components $\chi_{(i)}$ to be equal, thus reducing the Yang-Mills field degrees of freedom 
to a single scalar field with a quadratic potential, which excludes its chaotic behavior A general treatment of diagonal Yang-Mills Bianchi class A spacetimes using an orthonormal frame approach and expansion normalized variables can be found in~\cite{YM05}. However, a lack of suitable renormalized matter variables has prevented so far to obtain a global dynamical systems formulation on a compact state-space suitable for asymptotic description of those models. 

The present formulation can be extended to more general Bianchi models, where the isotropic case treated here appears as a special invariant set. 

 
\section*{Acknowledgments}
The authors thank FCT project PTDC/MAT-ANA/1275/2014 and CAMGSD, IST, Univ. Lisboa, through FCT project UIDB/MAT/04459/2020. FCM and VB thank CMAT, Univ. Minho, through FCT project Est-OE/MAT/UIDB/00013/2020 and FEDER Funds COMPETE. VB thanks FCT for the Ph.D. grant PD/BD/142891/2018.

\end{document}